\def \MyDate {November 6, 2025}
\def \MyVersion {V18}

\documentclass[a4paper, 11pt]{article}

\usepackage[super]{cite}
\usepackage{palatino, url, multicol}
\usepackage[english]{babel}
\usepackage{times}
\usepackage[T1]{fontenc}
\usepackage[utf8x]{inputenc}
\usepackage{graphics} 
\usepackage{epsfig} 
\usepackage{amsfonts}
\usepackage{amsmath,amsthm} 
\usepackage{amssymb}  
\usepackage{xcolor}
\usepackage{tikz}
\usetikzlibrary{arrows.meta,calc}
\usetikzlibrary{decorations.pathmorphing}
\tikzstyle{block}=[draw opacity=0.7,line width=1.4cm]
\usepackage{etoolbox}
\usepackage{color}
\usepackage[right]{lineno}

\newcommand{\pdiff}[3]{   \if 1#1     \frac{\partial #2}{\partial #3}      \else   \frac{\partial^{#1} #2}{\partial #3^{#1}}\fi     }

\DeclareMathOperator*{\argmin}{arg\,min}

\newtheorem*{theorem}{Theorem}
\newtheorem*{remark}{Remark}
\newtheorem*{proposition}{Proposition}

\newtheorem*{definition}{Definition}

\newtheorem*{lemma}{Lemma}

\title{\LARGE \bf The nature of mathematical models}
\author{Andrea De Gaetano $^{1,2,3,4, \dagger}$ \\
    \small $^{1}$CNR-IASI, Consiglio Nazionale delle Ricerche, Istituto di Analisi dei Sistemi ed Informatica, Rome, Italy;\\
    \small $^{2}$ CNR-IRIB, Consiglio Nazionale delle Ricerche, Istituto per la Ricerca e l’Innovazione Biomedica, Palermo, Italy;\\
    \small $^{3}$ Department of Biomatics, \'Obuda University, Budapest, Hungary;\\
    \small $^{4}$ Department of Mathematics, Mahidol University, Bangkok, Thailand.\\
    \small $^{\dagger}$Correspondence: andrea.degaetano@cnr.it
}
\date{\MyDate \, for arXiv: \MyVersion}

\begin{document}

\maketitle

\bigskip


\bigskip

\bigskip

\bigskip

\section*{Abstract}
    Mathematical modeling has become pervasive in applications, not only in physics or economy, but also in biomedicine and other ``soft'' sciences. To the conceptual formulation of a model there often follows its identification by statistical parameter estimation, given available observations. While the nature of the modeling process as well as its relationship with the attending statistical computations could both appear obvious to the practitioner, it may be useful to formalize them in a precise way. Insight into the process of (linear and nonlinear) model parameter estimation can be obtained from the description of the geometry of estimation in case space. The objective then is to also describe the geometry of modeling in the abstract, and to show how the correspondence between the conceptual context and the computational context can be formally represented.

\section{Introduction}

Mathematical modeling is central to defining and solving problems in the physical sciences, engineering or economics. Over the past few decades it has found progressively more widespread applications also in biomedicine and other \textit{``soft''} sciences. In the following we will make use, just by way of example, of naively simple problems in biology: while this choice stems from the background of the author as a biomedical modeler, the mathematical ideas are the same for any of the above application fields.\\

When considering a concrete problem in any field, the mathematical modeler is usually confronted with two distinct tasks: the first consists of imagining a potentially useful functional relationship among quantities, amenable to analysis and comparison with the real-world (the \underline{modeling} phase proper); the second consists in attempting to quantify the model parameters, leading to some hopefully good representation of the main features of the considered real-world phenomena (the model \underline{identification} phase). In practice, these tasks may succeed each other cyclically, an abstract formulation leading to an identification attempt which, if reputed unsuccessful, leads to a further formulation and so on. \\

The initial, modeling phase thus consists in imagining possible relationships among some quantitative features of interest of the reality under study, relationships usually expressed as algebraic or differential equations of various kinds. The initial concern of the present work is to understand exactly what the variables in these equations represent. Suppose for example we wished to estimate weight from height, and we were to write, as a very first attempt to express weight as a function of height, the equation
\begin{equation}
    W = k H
    \label{eq.WkH}
\end{equation}

Notice that at the present stage we are not dealing with any numbers at all, we are at the first stage of idealization of a relationship between a generic concept of height with a generic concept of weight. Our basic question is then: ``What, in the above equation, is $H$?'' The immediate, na\"ive answer would be ``Height!''. Yes, clearly, but whose height? Obviously $H$ is not meant to represent the numerical value of the height of any single individual, to be put in relationship with the weight $W$ of that same individual: changing individuals would then in general change the value of the parameter $k$, which is instead intended to express some general proportion between the height of people and their weight. Similarly, Eq. \ref{eq.WkH} does not express the relationship between an algebraic variable $W$ and an algebraic variable $H$: we do not imagine people's weights to follow exactly that functional relationship. Rather, we might consider $H$ as indicating a quantity (height) which \textbf{could} be measured on some individual or on a group of individuals, quite possibly with variability due to heterogeneity among people and/or to measurement error: while we have not yet measured anything at the present stage of the modeling exercise, we could do so in the future. We are thus led to regard $H$ as a random variable.\\

At this point we consider what $W$ might be. Notice that, if we were hypothetically to reflect on the observable weight of individuals and, following the previous line of reasoning, we were to consider such observable weight of individuals as a random variable, then this random variable (let us denote it provisionally by $W^o$) would \textbf{not} be the same as $W$. That $W$ is a random variable follows formally from the fact that it is a function of some numerical value $k$ and of another random variable $H$, so $W$ is indeed a random variable. But $W$ is that random variable obtained from $H$ through some sort of mathematical manipulation, whereas $W^o$ is another random variable, representing some characteristic of the people whom we may study in the future, characteristic which may or may not depend on $H$ and, even if it did depend on $H$, it could well in general not depend on it through Eq. \ref{eq.WkH}. In other words, we might consider $W^o$ as the objective (hence the $^o$) random variable that we are trying to approximate through some other random variable $W$, which is itself computable from $H$. We are thus led to the definition of a model, which will be formalized in the following, as an operator yielding a random variable (in this case $W$), with which we try to approximate another, objective random variable (in this case $W^o$). As will be made more precise in the following, ``approximation'' has meaning within a space endowed at least with a norm, while other useful concepts, such as projections, will need further structure (inner products), thus leading us to deal with a Hilbert space of random variables as the natural setting for the abstract modeling effort.\\

If we now consider the second, identification phase of  modeling practice, we see that it is commonly performed through the minimization of a cost functional (\textit{e.g.} by Ordinary Least Squares, OLS, or more generally by Maximum Likelihood, ML). Identification of course depends on having sampled some individuals from the population of interest and having observed on them both the intended criterion or target variable and the intended predictor variable(s) (such as, in the previous example, both  $W^o$ and $H$, from which last realizations of $W$ may be calculated). This procedure, which is inherently computational, is executed within an $n$-dimensional Euclidean space, referred to here as the \textit{case space}. \\

While both phases of the model construction process, modeling and identification, have been conducted innumerable times over centuries and are being more and more intensively conducted by applied mathematicians all over the world, to the best of the author's knowledge the relationship between abstract model building and concrete statistical estimation computations has so far been only intuited as their "obvious" correspondence. The aim of the present work is to make formally precise the relationship between the ideal construction of a mathematical model in the space of random variables and the statistical estimation procedure in $n$-dimensional case space.\\
It should be noted that the consideration of the Hilbert space $\mathcal{H}$ of (finite-variance) random variables is not new, see for example Small \cite{Small1994}. However, subsequent developments have centered mainly on the characterization of conditional expectation as the projection onto a linear subspace of $\mathcal{H}$ \cite{Bobrowski2013}, generated by an appropriate sub-$\sigma$-algebra of the original probability space. The present work, instead, focuses on the modeling process itself, giving rise to a generally nonlinear model manifold, and on the relationship of this construction with the (geometrical interpretation of) commonly performed statistical parameter estimation procedures.\\
The following discussion links more closely with the subject matter of information geometry \cite{Amari2009, Amari2010, BarndorffNielsen1986}, which deals with manifolds of probability distributions. While there exists a correspondence between (equivalence classes of) random variables and their induced probability distributions, and while differential geometric methods in information geometry help establish properties (\textit{e.g.} of convergence) of families of probability distributions, the focus here is on the direct expression of the model structure itself as a manifold in a space of random variables. In other words, while there exists a correspondence, the spaces dealt with in information geometry and here are different (one a space of distributions, the other a space of variables).\\

In the following, notation is first established. Then, supporting (well known) theorems are recalled and the formal constructions of a model, of a model function and of a model manifold are introduced. The geometry of statistical estimation, in the simple OLS case, is recalled. The problem is then tackled of how the challenges introduced by nonlinearity of the model in the two settings make the introduction of appropriate tangent spaces desirable. The mapping of the geometry in the random variable space to the geometry of the statistical estimation case space is finally formally defined. A brief final section recaps the results obtained, clarifies the relationship between model manifolds, conditional expectation and information geometry and points to possible further developments.

\section{Methods}

\subsection{Notation}

\subsubsection{\underline{Standard notation}}

Scalars are represented as (Latin or Greek) lower-scale plain-face letters ($a$, $\alpha$), vectors as lower-scale bold letters ($\boldsymbol{b}$, $\boldsymbol{\beta}$). However, when treating generic quantities, which may be either vector- or scalar-valued, we will often use plain-face rather than bold characters.

\smallskip

Matrices and higher-order arrays are written as (capital) boldface latin letters ($\boldsymbol{H}$ or $H$ when clear from the context). An array may show its dimensions as subscripted integers separated by a cross ($A_{n \times m}$ is the matrix  $A$ with $n$ rows and $m$ columns); an array element is identified by its indices (row $r$, column $c$) as $a_{r,c}$ or $a_{rc}$ when no confusion is possible. The $r^{th}$ row and $c^{th}$ column of the matrix $A$ are indicated respectively as $a_{r.} $  and $a_{.c}$ \,\,.  Notice the difference between $\boldsymbol{x}_i$ (the $i^{th}$ vector of a sequence), and  $x_i$ (the $i^{th}$ scalar of a sequence, such as the $i^{th}$  element of vector $\boldsymbol{x}$).

\smallskip

In order to distinguish between a random variable and its realization, we typically indicate the (scalar- or vector- valued) random variable upper case letter ($U$, $\boldsymbol{U}$) and its realization as the appropriate lower case letter ($u$, $\boldsymbol{u}$).

\smallskip

In order to keep notation consistent we will strive to have indices and dimensions correspond throughout, with  $i=1,\dots,n$  referring to cases (observations, criterion values, dependent variable values), $j=1,\dots,q$  referring to parameters, and $k=1,\dots,m$  referring to predictors (independent variables).

\smallskip

The real numbers are denoted by $\mathbb{R}$. Non-negative reals are denoted by $\mathbb{R}^+$. Spaces or manifolds are generally indicated by script capital letters like  $\mathcal{V}$,  $\mathcal{H}$, $\mathcal{M}$ or $\mathcal{Y}$.

\smallskip

Transposition of an array is denoted by superscript capital $T$, like in  $J^\intercal$. An array, of appropriate dimensions, having all elements equal to $0$  (or equal to $1$) is indicated as  $\boldsymbol{0}$   (or as  $\boldsymbol{1}$ ). The identity matrix of appropriate dimensions is indicated by  $\boldsymbol{I}$ .

\smallskip

The space spanned or generated by a set of vectors is indicated as \\$\mathcal{S} ( {\boldsymbol{x}_1, \boldsymbol{x}_2,\dots, \boldsymbol{x}_m} )$; when applied to a matrix $A$ the span function  $\mathcal{S} (A)$ indicates the space generated by the column vectors of the matrix.

\smallskip

The null space or kernel of a matrix $\boldsymbol{A}$ is denoted by $Ker(\boldsymbol{A})$,\,\,
$Ker (\boldsymbol{A}) = \{\boldsymbol{v} : \boldsymbol{A v} = 0\}$.

\smallskip

The orthogonal complement to a linear subspace $\mathcal{V}$ is denoted by  $\mathcal{V}^\perp$ ("\textit{$\mathcal{V}$-perp}"),\,\,
$\mathcal{V}^\perp = \{ \boldsymbol{x}: \boldsymbol{x}^\intercal  \cdot \boldsymbol{v} = 0 \ \  \forall \boldsymbol{v} \in \mathcal{V} \} $.

\smallskip

The generalized inverse of a matrix $\boldsymbol{A}$  is indicated as  $\boldsymbol{A}^-$  and satisfies
$\boldsymbol{A} \boldsymbol{A}^- \boldsymbol{A} = \boldsymbol{A}$  .

\smallskip

Probability measure is indicated by a capital  $P$, while the probability density function (pdf) of a random variable is indicated by a lower case  $p$ ; in this way, the probability of an event $A$ is
\begin{equation*}
    P(A) = \int_{A \in \Omega} dP
\end{equation*}

and if $B$ is a measurable set in $\mathbb{R}$ with respect to the usual Borel $\sigma$-algebra $\mathcal{B}$, and  $X$ is a random variable from $\Omega$ to $\mathbb{R}$ with density $p_X(x)$ then
\begin{equation*}
    P_X(B) := \int_{X^{-1}(B) \in \Omega} dP  =  \int_{B} p_X(x) dx
\end{equation*}

\smallskip

Under these conditions the expectation of a function $g(X)$ of the random variable $X$ is
\begin{equation*}
    E_{p_X}[g(X)] = \int_{-\infty}^{\infty} g(x)p(x)dx
\end{equation*}

\smallskip

Often used abbreviations are   \ \ WRT \ \ (With Respect To)   and  \ \ WLOG \ \ (Without Loss Of Generality).

End of proofs are indicated with the square box $\qed$ .

\smallskip

\subsubsection{\underline{Further notation}}

So far, the notation reflects standard statistical practice. In what follows, however, some necessary further notational conventions are established, which will make it easier to develop the argument in a precise way.

\smallskip

We let  $\boldsymbol{x}^o = [x^o_i]_{n \times 1}$ be a vector of observed ("o" stands for observed) realizations of a dependent variable (criterion), possibly corresponding to arbitrarily assigned or measured (assumedly without error) \underline{values} of  $m$  independent variables (predictors)  $\boldsymbol{U} = [u_{ik}]_{n \times m} = [\boldsymbol{u}_{i.}]_{n \times m}$ , one of which is often time $t$. We also let  $\boldsymbol{\theta} = [\theta_j]_{q \times 1}$  be a (vector-valued) parameter,  $\boldsymbol{\theta} \in \boldsymbol{\Theta} \subset \mathbb{R}^q$, and we suppose that the (generally nonlinear) functional form of the relationship between criterion and predictors is

\begin{equation*}
    x^o_i = x(\boldsymbol{u}_i,\boldsymbol{\theta}) + \varepsilon_i
\end{equation*}

or

\begin{equation*}
    \boldsymbol{x}^o = \boldsymbol{x}(\boldsymbol{U},\boldsymbol{\theta}) + \boldsymbol{\varepsilon}
\end{equation*}

where $\boldsymbol{\varepsilon}$ is an $n \times 1$ vector of "errors", \textit{i.e.} of casually, randomly occurring differences between modeled and observed values of the criterion for the given predictors. Notice that while $\boldsymbol{x}^o$ is a column array of values, $x()$ is a scalar function and $\boldsymbol{x}()$ is a vector-valued function.

\smallskip

In much of our discussion we will not need to indicate the predictors explicitly, since they are part of the experimental design and are fixed. With a slight abuse of notation we will thus simply write our predictor function as

\begin{equation*}
    \boldsymbol x = \boldsymbol{x}(\boldsymbol{\theta}) \quad \left ( = \boldsymbol{x}(\boldsymbol{U},\boldsymbol{\theta}) \right )
\end{equation*}

\smallskip

We indicate with a "hat" the estimated value of the parameter, $\hat{\boldsymbol{\theta}}$, to which corresponds a "best" prediction or forecast  $\hat{\boldsymbol{x}} = \boldsymbol{x(\hat{\boldsymbol{\theta}})}$. We also suppose that, the postulated relationship being the true one, there exists a "true" value of the parameter, $\boldsymbol{\theta}^*$, to which corresponds the true state of the system $\boldsymbol{x}^*= \boldsymbol{x}(\boldsymbol{\theta}^*)$.

\smallskip

For any value of the parameter we may compute the first derivative of the forecast with respect to the parameter, the Jacobian:

\begin{equation*}
    \boldsymbol{J} = \left [ \frac{\partial x_i(\boldsymbol{\theta})}{\partial \theta_j}  \right ] _{n \times q}
\end{equation*}

and, more specifically, we compute the Jacobian at the estimated parameter value as

\begin{equation*}
    \hat{\boldsymbol{J}} = \left[ \frac{\partial x_i(\hat{\boldsymbol{\theta}})}{\partial \theta_j}  \right ] _{n \times q}  = \left.  \frac{\partial x_i(\boldsymbol{\theta})}{\partial \theta_j}  \right|_{\hat{\boldsymbol{\theta}}}
\end{equation*}

\smallskip

Using an alternative notation for the Jacobian we may write (in the spirit of Seber \cite{Seber2003})

\begin{equation*}
    \boldsymbol{\overset{.}{x}} = \boldsymbol{J} =  \frac{\partial x_i(\boldsymbol{\theta})}{\partial \theta_j}
\end{equation*}

so that the three-dimensional array of second-order partial derivatives is denoted as

\begin{equation*}
    \boldsymbol{\overset{..}{x}} = \left [ \frac{\partial^2 x_i(\boldsymbol{\theta})}{\partial \theta_j \ \partial \theta_k} \right ]_{q \times n \times q}
\end{equation*}

\medskip

with the obvious meaning being attributed to $\hat{\boldsymbol{\overset{.}{x}}}$,  $\hat{\boldsymbol{\overset{..}{x}}}$,  $\boldsymbol{\overset{.}{x}}^*$ and  $\boldsymbol{\overset{..}{x}}^*$ .

Multiplications involving these three-dimensional arrays are defined, for the scope of the present work, simply as the slice-wise equivalent: the array  $\boldsymbol{A}_{m \times n \times q}$ is treated as a pile of  $n$  slices, such that the $i$-th slice is an  $m \times q$  matrix   $(\boldsymbol{A}_{.i.})_{m \times q}$ . Then pre-multiplication by a matrix  $\boldsymbol{B}_{r \times m}$ and post-multiplication by a matrix  $\boldsymbol{C}_{q \times s}$ are defined as producing  $n$-piles of  matrices  $(\boldsymbol{B A}_{.i.})_{r \times n \times q}$  or  $(\boldsymbol{A}_{.i.} \boldsymbol{C})_{m \times n \times s}$ respectively, with the understanding that dimensions collapsing to $1$ disappear, so that, for example,  we treat   $\boldsymbol{v}_{1 \times m} \boldsymbol{A}_{m \times n \times q}$  simply as the matrix  $(\boldsymbol{vA})_{n \times q}$.

\begin{itemize}
    \item   a \textbf{bold} symbol indicates a vector. Thus $\boldsymbol{\varepsilon}$, $\boldsymbol{\xi}$ and $\boldsymbol{e}$ are vectors.
    
    \item A \underline{Hamel basis} of a vector space $\mathcal{V}$ over a field $\mathbb{K}$ is a set $\mathcal{B} \subset \mathcal{V}$ such that every element $v \in \mathcal{V}$ can be written uniquely as a \emph{finite} linear combination of elements of $\mathcal{B}$. This means that $\mathcal{B}$ is both linearly independent and spans $\mathcal{V}$ in the sense that for every $v \in \mathcal{V}$, there exist finitely many $b_i \in \mathcal{B}$ and scalars $\alpha_i \in \mathbb{K}$ such that $v = \sum_{i=1}^{n_v} \alpha_i b_i$. 
    
    \item In finite-dimensional spaces (such as $\mathbb{R}^d$) the standard basis with $d$ elements is a Hamel basis. However, in infinite-dimensional spaces, Hamel bases are typically uncountable and non-constructive, relying on the Axiom of Choice for their existence. 
    
    \item A \underline{Schauder basis} allows infinite (typically countable) linear combinations. It must be defined in the context of topological vector spaces, where convergence of the series is required. 

    \item An orthonormal basis for an $n$-dimensional Hilbert space (such as $\mathbb{R}^n$) is a collection of $n$ unit-length, orthogonal basis vectors, say $\boldsymbol H = \{ \boldsymbol{\eta}_1,\dots,\boldsymbol{\eta}_n\}$ or $\boldsymbol P = \{ \boldsymbol{\rho}_1,\dots,\boldsymbol{\rho}_n\}$.

    \item a vector $ \boldsymbol{v}$ is the sum of its \underline{components} $\{ \boldsymbol{v}_i\}_{i=1}^n$, each of which can be expressed as a \underline{coefficient} multiplying a basis vector, say $\boldsymbol{v}_i = \text{v}_{\eta}^i \cdot \boldsymbol{\eta}_i$, where we emphasize that the coefficient $\text{v}_\eta^i$ refers to the chosen basis $\boldsymbol H$  \\ {\footnotesize (Notice that we are NOT using here Einstein's summation notation, the dot above denotes scalar multiplication in the vector space $\mathbb{R}^n$ over the reals)}

    \item we may collect the coefficients into a column array of coefficients, relative to the chosen basis, say ${\text{\textbf{v}}}_\eta = \{ \text{v}^1_\eta, ..., \text{v}^n_\eta\}^\intercal$ \,\, or \,\, ${\text{\textbf{v}}}_\rho = \{ \text{v}^1_\rho, ..., \text{v}^n_\rho\}^\intercal$

    \item we can thus concisely express a vector as its expansion with respect to a given basis: \,\, $\boldsymbol{v} = \boldsymbol{E} {\text{\textbf{v}}}_\eta = \text{v}_{\eta}^i \boldsymbol{\eta}_i$ \,\, or \,\, $\boldsymbol{v} = \boldsymbol{P} {\text{\textbf{v}}}_\rho = \text{v}_{\rho}^i \boldsymbol{\rho}_i$ \\ {\footnotesize (Notice that here we ARE using Einstein's summation notation)}

\end{itemize}

\subsection{Geometry of Estimation in Case Space ($\mathbb{R}^n$)}

We summarize here the standard geometrical interpretation of Ordinary Least\\ Squares parameter estimation. While none of the material in this section is new\cite{Bates1988, Seber1989, Seber2003}, the following description should help clarify the correspondence of this with subsequently presented material.\\

When talking about a model in an applied setting we  typically employ the ``usual representation''

\begin{equation}
    x = x(u,\boldsymbol{\theta})
\end{equation}

where a dependent variable $x$ is assumed to depend upon some independent predictor $u$ (often indicating time in dynamical systems) via some functional form and some $q$-dimensional parameter $\boldsymbol{\theta}$. This relationship is typically graphed showing curves of $x$ as a function of $u$, different curves corresponding to different values of $\boldsymbol{\theta}$. Since by changing the parameter value the prediction $x(u,\boldsymbol{\theta})$ obviously changes, this clearly leads to thinking about estimation as adapting the predicted curve by minimizing some functional of the ``errors'' (differences between predicted and observed values of $x$), yielding \,\, $\hat {\boldsymbol{\theta}}$ \,\, as the \,$\argmin$\, of the functional.\\

However, the values $u$ takes are fixed (by design or by sampling), and we may drop $u$ from notation, since what we are really interested in is the relationship between parameter and prediction. Also, we may indicate with $\boldsymbol{x}$ the set of predicted values that $x$ takes for each of $n$ experimental units (corresponding to a given set of predictors) as $\boldsymbol{\theta}$ varies:

\begin{equation}
    \boldsymbol{x} = \boldsymbol{x}(\boldsymbol{\theta})
    \label{Eq:xbytheta}
\end{equation}

and we may also indicate with  $\boldsymbol{x}^o$ (observed) the actually measured or observed values of $x$ in the $n$ experimental units. We find thus useful to work in  ($n$-dimensional)  $\boldsymbol{x}$-space or \underline{case space}, $\mathbb{R}^n$. Case space is therefore defined as the $n$-dimensional Euclidean space where each axis represents one of the experimental units, and where along each axis we may identify both the actually observed value for that unit and any predicted value for that unit (depending on the chosen value for the parameter $\boldsymbol{\theta}$). In other words, case space is the $n$-dimensional space where all possible values of $\boldsymbol{x}$, either measured or predicted, referring to the $n$ experimental units can be represented. The appeal of working in case space consists in an immediate visualization of the geometry of the problem: the whole observed sample realization is a single point $\boldsymbol{x}^o$ in $n$-dimensional case space, and we may trace a (generally nonlinear) $q$-dimensional \underline{prediction surface}  $\mathcal{M}_{\boldsymbol{x}}$ as the set of values that the predicted $\boldsymbol{x}$ can take as $\boldsymbol{\theta}$ varies. Indeed, if $\boldsymbol{\theta} \in \mathbb{R}^q$ and $x(\boldsymbol{\theta})$ is sufficiently regular then the prediction surface ${\mathcal{M}_{\boldsymbol{x}}}$ is a $q$-dimensional manifold embedded in case space ($\mathbb{R}^n$). Notice that on this manifold we can have a ``canonical'' global chart map
$\boldsymbol{x}^{-1}: \mathbb{R}^n \rightarrow \mathbb{R}^q$,   $\boldsymbol{x}^{-1}:\boldsymbol{x}=\boldsymbol{x}(\boldsymbol{\theta}) \in \mathbb{R}^n  \mapsto \boldsymbol{\theta} \in \mathbb{R}^q$.\\

A number of useful remarks may be offered at this point:

\begin{itemize}
    \item[\textbullet]  $\mathbb R^n$ has the structure of a vector space over the reals
    \item[\textbullet]  we can define on $\mathbb R^n$ the Standard Euclidean topology $\mathcal{O}$
    \item[\textbullet]  $\mathcal{O}$ generates the Borel $\sigma$-algebra $\mathcal{B}$
    \item[\textbullet]  $(\mathbb R^n,\mathcal{B})$ is a measurable space
    \item[\textbullet]  we have the Lebesgue measure $\lambda$ defined on $\mathcal{B}$ with values in $\mathbb R^+ := \{r \in \mathbb R, r \geqslant 0\}$, making $(\mathbb R^n,\mathcal{B}, \lambda)$ into a measure space
    \item[\textbullet]  we can endow $\mathbb R^n$ with a norm, the usual euclidean norm $\|\boldsymbol{x}\| := \sqrt{\sum_{i=1}^n x_i^2}$
    \item[\textbullet]  with respect to $\|\boldsymbol{x}\|$ the vector space  $\mathbb R^n$ is complete, hence Banach
    \item[\textbullet]  we can define on $\mathbb R^n$ an inner product compatible with the norm, \quad $\langle \boldsymbol{x},\boldsymbol{y} \rangle := \sum_{i=1}^n x_i y_i$, \quad making $\mathbb R^n$ Hilbert
\end{itemize}

\bigskip

\subsubsection{\underline{Linear relationship in case space: geometry}}

We begin by hypothesizing that the relationship given by Eq.\ref{Eq:xbytheta} between $\boldsymbol{x}$ and $\boldsymbol{\theta}$ is actually linear. Referring to Fig.\ref{fig:LinearCaseSpace} at page \pageref{fig:LinearCaseSpace} we can now identify a number of relevant geometric features of the problem:

\begin{itemize}
    \item ${ \boldsymbol{\theta}}$ is a $q$-dimensional parameter, $ \boldsymbol{\theta} \in \Theta \subset \mathbb{R}^q$
    \item $\boldsymbol{x} = \boldsymbol{x}(\boldsymbol{\theta})$ is a function mapping (in this case linearly) the set \, $\Theta$ \, of allowable parameter values into $\mathbb{R}^n$, $\boldsymbol{x} : \Theta \rightarrow \mathbb{R}^n$
    \item ${\mathcal{M}_{\boldsymbol{x}}}$ (the \underline{prediction surface}) is the set of possible values taken by the function $\boldsymbol{x}(\boldsymbol{\theta})$; in this linear case it is a $q$-dimensional subspace of $\mathbb{R}^n$
    \item since we assume by hypothesis the relationship given by Eq.\ref{Eq:xbytheta}, then  $\boldsymbol{x}^*$ is the \underline{state}, the supposedly true value of the observed phenomenon, necessarily lying on the supposedly true prediction surface, at a position indexed by the supposedly true parameter value $\boldsymbol{\theta}^*$
    \item $\boldsymbol{x}^o$ is the \underline{observation}, the observed point in case space, assumed to be generated by the true linear relationship, but to be affected by some observation error
    \item $\hat{\boldsymbol{x}}$ is our best estimate for the state, the point on ${\mathcal{M}_{\boldsymbol{x}}}$ closest to the observation $\boldsymbol{x}^o$, indexed by the parameter value $\hat {\boldsymbol{\theta}}$; notice that it makes sense to speak about  \textit{closeness} since, as mentioned before, $\mathbb{R}^n$ is normed; also, we are arbitrarily taking minimum distance between $\boldsymbol{x}^o$ and ${\mathcal{M}_{\boldsymbol{x}}}$ as the basis for our estimation (OLS)
    \item $\boldsymbol{{\varepsilon}}^o = \boldsymbol{x}^o - \boldsymbol{x}^*$ is the \underline{error} made by observing $\boldsymbol{x}^o$ instead of $\boldsymbol{x}^*$, the difference between the observation and the true state; it is the actual realization of the random variable ${\boldsymbol{\varepsilon}}$:  since we do not know $\boldsymbol{x}^*$ we do not know ${\boldsymbol{\varepsilon}}^o$, even though we might make hypotheses about how ${\boldsymbol{\varepsilon}}$ should be distributed
    \item ${\boldsymbol{e}}^o = \boldsymbol{x}^o - \hat{\boldsymbol{x}}$ is the \underline{residual}, the observed difference between the observation $\boldsymbol{x}^o$ and our best guess $\hat{\boldsymbol{x}}$; once again,  ${\boldsymbol{e}}^o$ is the actual realization of the random variable ${\boldsymbol{e}}$
    \item ${\boldsymbol{\xi}}^o = \hat{\boldsymbol{x}} - \boldsymbol{x}^*$ is the actual \underline{flaw} (we cannot qualify it as ``observed'' flaw, since we do not really observe it) 
    \footnote{ The term ``flaw'' is here introduced, since a name for $\boldsymbol{\xi}$ could not be found in the standard literature}, the assessment mistake we would make by placing our best guess at $\hat{\boldsymbol{x}}$ instead of at $\boldsymbol{x}^*$ and, as before, it is the realization of the random variable $\boldsymbol{\xi}$.
\end{itemize}

\begin{figure}
    \includegraphics[height=7cm]{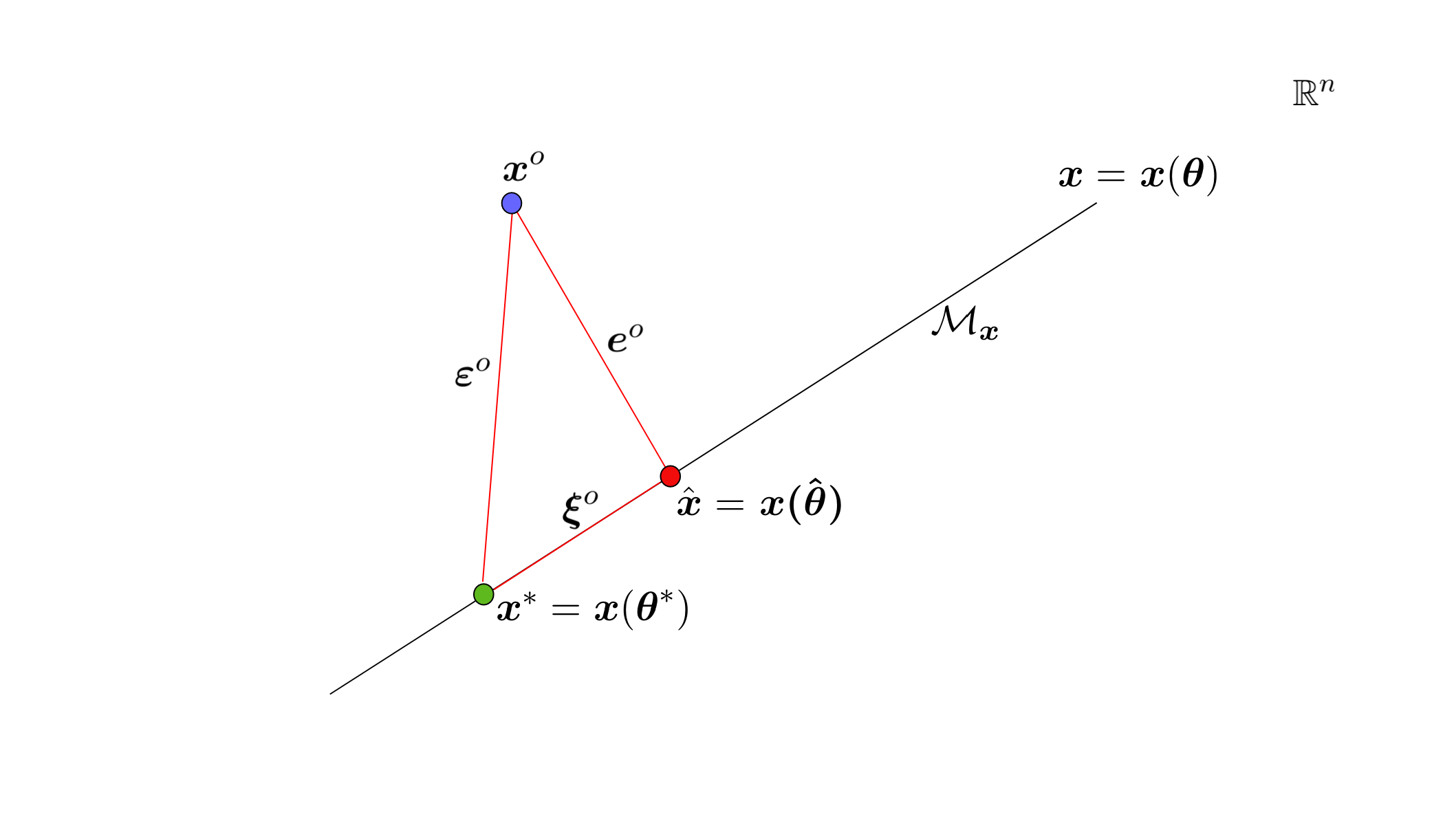}
    \caption{The geometry of a linear relationship in case space}
    \label{fig:LinearCaseSpace}
\end{figure}

\subsubsection{\underline{Linear relationship in case space: estimation by projection}}

Given the above geometry, we can now describe a possible procedure for estimating $\boldsymbol{\theta}$. Notice first that, since ${\mathcal{M}_{\boldsymbol{x}}}$ is a (nonempty) subspace of $\mathbb R^n$, it is convex and a Hilbert space in its own right. From the Hilbert projection theorem\cite{Rudin1987} a unique projection operator $\mathcal P_{\mathcal{M}_{\boldsymbol{x}}} : \mathbb R ^n \rightarrow {\mathcal{M}_{\boldsymbol{x}}}$ can be defined. In fact, if it were $\boldsymbol{x} = U \boldsymbol{\theta}$ for some matrix $U$, we could write $\mathcal P_{\mathcal{M}_{\boldsymbol{x}}} = U(U^\intercal U)^{-1} U^\intercal$.
Since this projection operator satisfies (always by the Hilbert projection theorem)
\begin{equation*}
	\| \boldsymbol{x} - \mathcal{P}_{\mathcal{M}_{\boldsymbol{x}}} \boldsymbol{x}\| < \| \boldsymbol{x} -  \boldsymbol{y}\|\,,\, \forall  \boldsymbol{y} \in {\mathcal{M}_{\boldsymbol{x}}},  \boldsymbol{y} \neq \mathcal{P}_{\mathcal{M}_{\boldsymbol{x}}}  \boldsymbol{x} \,,
\end{equation*}

we have that
$ \hat {\boldsymbol{\theta}} = \underset{\boldsymbol{\theta} \in \boldsymbol{\Theta}}{\argmin} (|| \boldsymbol{x}(\boldsymbol{\theta}) -  \boldsymbol{x}^o||^2)$ , \quad where $\boldsymbol{\Theta}$ is the admissible set of values of $\boldsymbol{\theta}$.\\
\noindent
The projection $\hat{\boldsymbol{x}} = \mathcal{P}_{\mathcal{M}_{\boldsymbol{x}}} \boldsymbol{x}^o$ of the observed sample $\boldsymbol{x}^o$ onto the prediction surface ${\mathcal{M}_{\boldsymbol{x}}}$  thus identifies the unique point $\hat {\boldsymbol{x}} \in {\mathcal{M}_{\boldsymbol{x}}}$
which is closest to $\boldsymbol{x}^o$ in the OLS
sense. The value $\hat {\boldsymbol{\theta}}$ such that $\hat {\boldsymbol{x}} = \boldsymbol{x}(\hat{\boldsymbol{ \theta}})$     is the desired (point) parameter estimate.\\

Under reasonable assumptions:
\begin{itemize}
    \item the relationship is true
    \item observations are determined by the true relationship, design variable(s) and the supposedly true parameter value $\boldsymbol{\theta}^*$, plus error
    \item errors on the observations of the $n$ experimental units are i.i.d. (independent identically distributed), with a normal distribution of zero mean and variance $\sigma^2$
\end{itemize}
\noindent we have:
\begin{enumerate}
    \item the OLS point estimate coincides with the Maximum Likelihood Estimate (with the attending nice properties of efficiency and consistency of the MLE);
    \item asymptotic confidence regions of $\hat {\boldsymbol{x}} = \mathcal{P}_{\mathcal{M}_{\boldsymbol{x}}}  \boldsymbol{x}^o$ on ${\mathcal{M}_{\boldsymbol{x}}}$, hence of $\hat {\boldsymbol{\theta}}$ can be obtained \cite{Seber2003}.
\end{enumerate}

\noindent
In fact under these hypotheses:
\begin{itemize}

    \item $\|\boldsymbol{\varepsilon}\|^2 = \sum_{i=1}^n \, \varepsilon_i^2$, where $\forall \, i \,\,\, \varepsilon_i\sim \mathcal{N}(0,\sigma^2)$  (Pythagoras theorem);

    \item under a rotation of the basis for $\mathbb{R}^n$ we can obtain a new basis with the first $q$ elements spanning ${\mathcal{M}_{\boldsymbol{x}}}$ and with the remaining $(n-q)$ elements orthogonal to ${\mathcal{M}_{\boldsymbol{x}}}$;

    \item the components of $\boldsymbol{\varepsilon}$ in the new basis are again i.i.d  $\mathcal{N}(0,\sigma^2)$ , which follows from the fact that an orthonormal linear transformation of a standard multivariate normal is a standard multivariate normal;

    \item since $\boldsymbol{\varepsilon} = \boldsymbol{\xi} + \boldsymbol{e}$, with $\boldsymbol{\xi} \perp \boldsymbol{e}$, we have again by Pythagoras theorem that $\|\boldsymbol{\varepsilon}\|^2 = \|\boldsymbol{\xi}\|^2  + \|\boldsymbol{e}\|^2 = $, where $\|\boldsymbol{\xi}\|^2 = \sum_{j=1}^n \, \xi_j^2$ and $\|\boldsymbol{e}\|^2 = \sum_{i=1}^n \, e_i^2$;

    \item notice however that with respect to the new basis\\
    $e_j = 0 \,\, \forall j \in \{1,\dots,q\}$ \,\,and\,\, $\xi_i = 0 \,\, \forall i \in \{q+1,\dots,n \}$  , while \\
    $\xi_j \sim \mathcal{N}(0,\sigma^2) ,\, j \in \{1,\dots,q \}$ \,\,and \,\, $e_i \sim \mathcal{N}(0,\sigma^2) ,\, i \in \{q+1,\dots,n \}$

\end{itemize}

From the above follow important statistical consequences:

\begin{itemize}

    \item $\frac{\|\boldsymbol{\varepsilon}\|^2}{\sigma^2}
    = \sum_{i=1}^n \, \frac{\varepsilon_i^2}{\sigma^2}
    = \sum_{i=1}^n \, \left ( \frac{\varepsilon_i}{\sigma
    } \right )^2$,\,\, where \,\, $ \frac{\varepsilon_i}{\sigma} \sim \mathcal{N}(0,1)$, \,\,so\\ $\frac{\|\boldsymbol{\varepsilon}\|^2}{\sigma^2} \sim \chi_n$\,\,{\footnotesize (Chi-square with $n$ degrees of freedom)}\\

    \item
    $\frac{\|\boldsymbol{\xi}\|^2}{\sigma^2} = \sum_{i=1}^n \, \frac{\xi_i^2}{\sigma^2} = \sum_{j=1}^q \, \left ( \frac{\xi_j}{\sigma} \right)^2  $, \,\,where\,\, $ \frac{\xi_j}{\sigma} \sim \mathcal{N}(0,1)$, \,\,so\\
    $\frac{\|\boldsymbol{\xi}\|^2}{\sigma^2} \sim \chi_q$\,\,{\footnotesize (Chi-square with $q$ degrees of freedom)}\\

    \item $\frac{\|\boldsymbol{e}\|^2}{\sigma^2} = \sum_{i=1}^n \, \frac{e_i^2}{\sigma^2} = \sum_{i=q+1}^n \, \left ( \frac{e_i}{\sigma} \right ) ^2 $, \,\,where\,\, $ \frac{e_i}{\sigma} \sim \mathcal{N}(0,1)$, \,\,so\\
    $\frac{\|\boldsymbol{e}\|^2}{\sigma^2} \sim \chi_{n-q}$\,\,{\footnotesize (Chi-square with $(n-q)$ degrees of freedom)}\\

    \item $\frac{\|\boldsymbol{\xi}\|^2}{\sigma^2}  /  \frac{\|\boldsymbol{e}\|^2}{\sigma^2}  = \frac{\|\boldsymbol{\xi}\|^2}{\|\boldsymbol{e}\|^2}  \sim \textit{F}_{q,(n-q)}$\,\,{\footnotesize (Fisher with $q$ and $(n-q)$ degrees of freedom)}\\

\end{itemize}

\noindent and we can compute a $(1-\alpha)$ confidence region for  $\|\boldsymbol{\xi}\|$. \\
In other words: under the stated hypotheses ($\varepsilon_i  $ i.i.d. $\sim \mathcal{N}(0,\sigma^2)$), taking into account the dimension $n$ of case space as well as the dimension $q$ of the prediction surface ${\mathcal{M}_{\boldsymbol{x}}}$, by observing the size of the residual from ${\mathcal{M}_{\boldsymbol{x}}}$, $\|\boldsymbol{e}\|$, we have information on the likely size of the flaw along ${\mathcal{M}_{\boldsymbol{x}}}$, $\|\boldsymbol{\xi}\|$.

\subsubsection{\underline{Nonlinear relationship in case space}}

\begin{figure}
	\includegraphics[height=7cm]{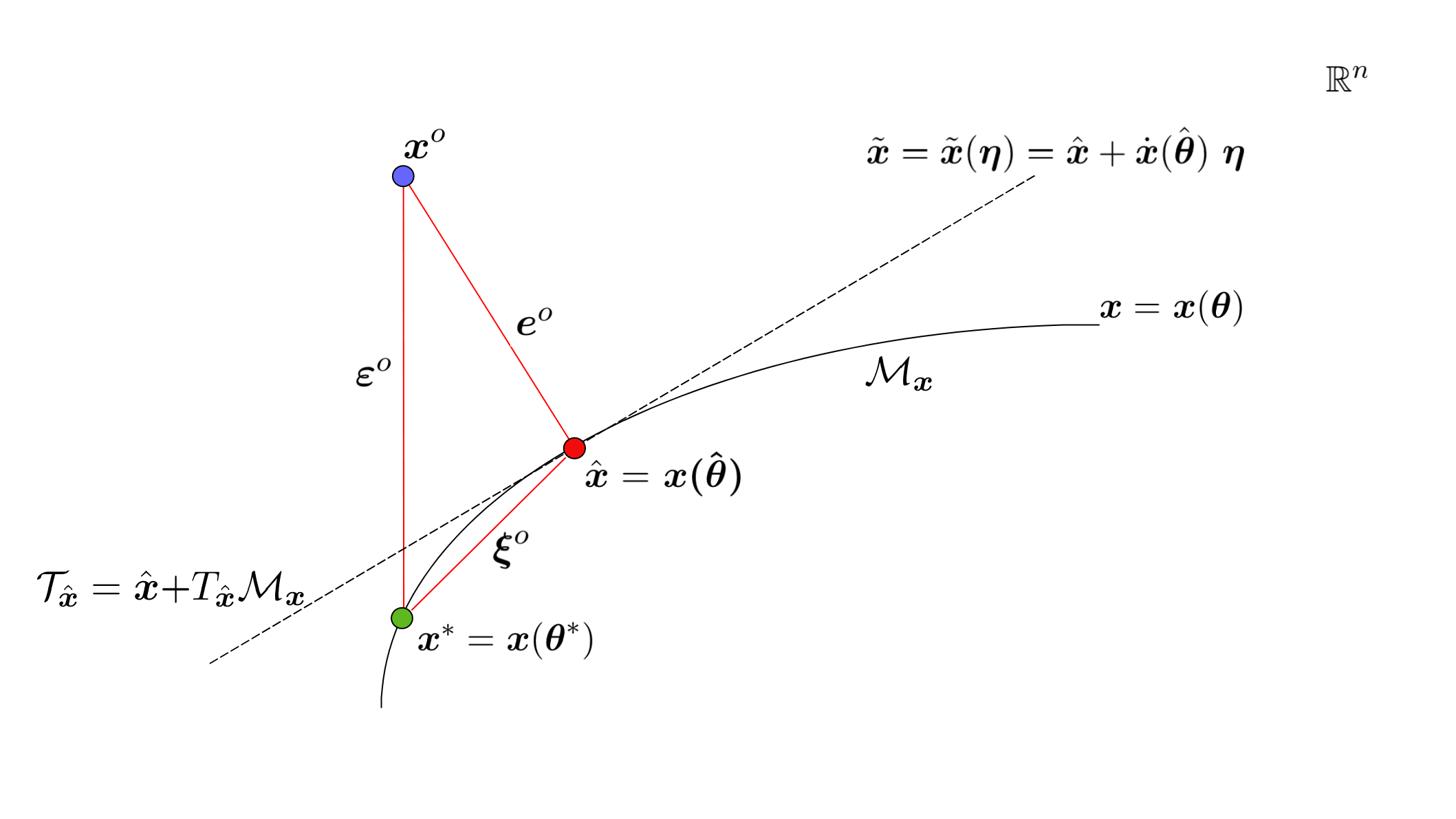}
	\caption{Nonlinear relationship in case space: original geometry}
	\label{fig:NonlinearCaseSpaceOriginalRn}
\end{figure}

By referring now to Fig.\ref{fig:NonlinearCaseSpaceOriginalRn} at page \pageref{fig:NonlinearCaseSpaceOriginalRn}  we see that a construction, similar to what was provided for the linear case, does not allow us to draw the same conclusions. The fact is that in this context ${\mathcal{M}_{\boldsymbol{x}}}$ is not a subspace of $\mathbb R ^n$ and in general it may not even be convex, hence the Hilbert projection theorem does not apply and a unique projection operator $\mathcal P_{\mathcal{M}_{\boldsymbol{x}}} : \mathbb R ^n \rightarrow {\mathcal{M}_{\boldsymbol{x}}}$ cannot be defined.\\

We may however still look for a value $ \hat {\boldsymbol{\theta}} = \underset{\boldsymbol{\theta} \in \boldsymbol{\Theta}}{\arg \min} (|| \boldsymbol{x}(\boldsymbol{\theta}) -  \boldsymbol{x}^o||)$ , \quad where again $\boldsymbol{\Theta}$ is the admissible set of values of $\boldsymbol{\theta}$, same as for the linear case. The existence, much less the uniqueness, of such a value $\hat {\boldsymbol{\theta}}$ would not be guaranteed, but we would hope an isolated minimum of the loss function $\ell(\boldsymbol{\theta}) = || \boldsymbol{x}(\boldsymbol{\theta}) -  \boldsymbol{x}^o||^2 = || \boldsymbol{e}(\boldsymbol{\theta})||^2$ would exist for some $\boldsymbol{x}(\boldsymbol{\theta})$ in the interior of a relevant open set \,\, $\mathcal{U} \subset {\mathcal{M}_{\boldsymbol{x}}}$\,: in this case, the value $ \hat {\boldsymbol{\theta}}$ such that $\hat {\boldsymbol{x}} = \boldsymbol{x} ( \hat {\boldsymbol{\theta}})$, the desired (point) parameter estimate, could be found numerically as \quad $ \hat {\boldsymbol{\theta}} = \underset{\boldsymbol{\theta} \in \boldsymbol{x}^{-1}(\mathcal{U}) \subset  \boldsymbol{\Theta}}{\arg \min} (|| \boldsymbol{x}(\boldsymbol{\theta}) -  \boldsymbol{x}^o||^2)$.\\

We would like to be able to make at least approximate confidence statements about the parameter estimate in the nonlinear case as well. To do so, and in case an isolated local minimum has been found, we might wish to use a local linear approximation $\mathcal{T}_{\hat{\boldsymbol{x}}} = {\hat{\boldsymbol{x}}} +  T_{\hat{\boldsymbol{x}}}{\mathcal{M}_{\boldsymbol{x}}}$ to the prediction surface ${\mathcal{M}_{\boldsymbol{x}}}$ around the estimated state $\hat{\boldsymbol{x}}$. Notice that \,$\mathcal{T}_{\hat{\boldsymbol{x}}}$\, is not in general a subspace, but rather an affine space.

The Hilbert projection theorem however extends naturally to affine subspaces, the key observation being that the problem can be reduced to projecting onto a linear subspace by translating the space appropriately. The projection onto an affine subspace is unique and satisfies an orthogonality condition relative to the associated direction subspace  \, $T_{\hat{\boldsymbol{x}}}{\mathcal{M}_{\boldsymbol{x}}}$ \,, see Appendix \ref{sec:HilbertAffine} for details. In this way, if an isolated minimum of $|| \boldsymbol{x}(\boldsymbol{\theta}) -  \boldsymbol{x}^o||^2$ at $\hat {\boldsymbol{x}}$ exists, then we can define a (locally unique) projection operator \,\, $\mathcal P_{\mathcal{T}_{\hat{\boldsymbol{x}}}} : \mathbb R ^n \rightarrow \mathcal{T}_{\hat{\boldsymbol{x}}}$, \,\, $ \mathcal P_{\mathcal{T}_{\hat{\boldsymbol{x}}}} : \boldsymbol{x}^o \mapsto \hat {\boldsymbol{x}}$\,.\\

\noindent
Under the same assumptions that we made for the linear case:
\begin{itemize}
    \item the relationship is true
    \item observations are determined by the true relationship, design variable(s) and the supposedly true parameter value $\boldsymbol{\theta}^*$, plus error
    \item errors on the observations of the $n$ experimental units are i.i.d. (independent identically distributed), with a normal distribution of zero mean and variance $\sigma^2$
\end{itemize}

\noindent we have that \cite{Seber1989,Bates1988}:
\begin{enumerate}
    \item there is no guarantee that the local OLS point estimate is the global optimum.

    \item approximate asymptotic confidence regions of $\hat {\boldsymbol{x}}$ on $\mathcal{T}_{\hat{\boldsymbol{x}}}$, hence of $\hat {\boldsymbol{\theta}}$ can be obtained: such regions would not be too different from the correct asymptotic confidence regions of $\hat {\boldsymbol{x}}$ on ${\mathcal{M}_{\boldsymbol{x}}}$ if the curvature of ${\mathcal{M}_{\boldsymbol{x}}}$ at $ \hat {\boldsymbol{x}}$ is not too large \cite{Panunzi2005}.
\end{enumerate}

\section{Geometry of Estimation in the Hilbert space $\mathcal{H}$ of random variables}

\subsection{Preliminaries}

From what has been presented in the Introduction, modeling appears to be the act of hypothesizing mathematical relationships among idealized quantities, some of which may be observed (\textit{i.e.} numerically measured in a certain number of cases). To make this idea precise, a few relevant, very standard results are listed in the following. Proofs can be found in any textbook of functional analysis or mathematical statistics \cite{Folland2013, Rohatgi1976}.

\bigskip

First of all, the notion of a quantity that might be observed in a certain number of cases  (with possibly different outcomes) coincides with the mathematical concept of a \textbf{random variable}.

\subsubsection*{\underline{Random variables}}
Let $(\Omega, \mathcal{F}, P)$ be a probability space.

Recall that an $n$-dimensional real random variable $U:\Omega\rightarrow\mathbb{R}^n$ is a \textit{measurable map} from the probability space $(\Omega,\mathcal{F},P)$ to the measurable space $(\mathbb{R}^n,\mathcal{B}(\mathbb{R}^n))$, where $\mathcal{B}(\mathbb{R}^n)$ is the Borel $\sigma$-algebra of $\mathbb R^n$. Recall also that, given a probability space $(\Omega, \mathcal{F},P)$ and a random variable $U:\Omega\rightarrow\mathbb{R}^n$, the $\sigma$-algebra generated by the r.v. $U$, $\sigma(U)$, is the smallest $\sigma$-algebra containing all sets $U^{-1}(A), A \in \mathcal{B}(\mathbb{R}^n)$.

\subsubsection*{\underline{The Doob-Dynkin lemma for random variables} \textit{(also see Fig.\ref{fig:DoobDynkinDiagram}, page \pageref{fig:DoobDynkinDiagram})} }

\begin{theorem}[Doob-Dynkin]
    Let $(\Omega, \mathcal{F},P)$ be a probability space, $(\mathbb{R}^n,\mathcal{B}_n)$ and $(\mathbb{R}^m,\mathcal{B}_m)$ measurable spaces, $U:\Omega\rightarrow\mathbb{R}^n$ a random variable. Then a function $X:\Omega\rightarrow\mathbb{R}^m$ is $\sigma(U)$-$\mathcal{B}_m$-measurable iff $\,\, \exists \, x:\mathbb{R}^n\rightarrow\mathbb{R}^m$, $x$ $\mathcal{B}_n$-$\mathcal{B}_m$-measurable, such that $X=x\circ U = x(U)$. In this case, consequently, $X$ is also a random variable.\\
\end{theorem}

\begin{figure}
    \begin{center}
        \begin{tikzpicture}
            \begin{scope}[every node/.style={draw=white!60,thick,draw}]
                \node (A) at (0,0) {$\Omega$};
                \node (B) at (2,0) {$\mathbb{R}^n$};
                \node (C) at (2,-3)  {$\mathbb{R}^m$};
            \end{scope}

            \begin{scope}[>={Stealth[black]},
                every node/.style={},
                every edge/.style={draw=black}]
                \path [->] (A) edge [above] node {$U$} (B);
                \path [->] (A) edge [left] node {$X = x \circ U$} (C);
                \path [->] (B) edge [right] node {$x$} (C);
            \end{scope}
        \end{tikzpicture}
    \end{center}
    \caption{By the Doob-Dynkin lemma this diagram commutes for measurable functions}
    \label{fig:DoobDynkinDiagram}
\end{figure}

\subsubsection*{\underline{The spaces  $\mathcal{L}^p(\Omega,\mathcal{F},P)$ and $L^p(\Omega,\mathcal{F},P)$}}

The set of measurable functions $X$ on $\Omega$ such that $(\int_{\Omega}|X|^p dP)^{\frac{1}{p}}<\infty$ is denoted by  $\mathcal{L}^p(\Omega,\mathcal{F},P)$. It can be easily seen that $\mathcal{L}^p$ is a vector space.\\
Let $Y, X \in \mathcal{L}^p(\Omega,\mathcal{F},P)$. If $Y$ and $X$ differ only on a set of $P$-measure zero (\textit{i.e.} $Y \overset{P a.e.}{=} X$) then $(\int_{\Omega}|Y - X|^p dP)^{\frac{1}{p}} = 0$. This leads naturally to consider such measurable functions as equivalent if they are equal almost everywhere. We can thus quotient $\mathcal{L}^p$ by the equivalence relation $ \overset{P a.e.}{=} $:
the quotient space $L^p = \mathcal{L}^p/\overset{P a.e.}{=}$ is defined to be the set of all equivalence classes under the relation $\overset{P a.e.}{=}$.\\
We define the equivalence class of a measurable function $X$ under the equivalence relationship $\overset{P a.e.}{=}$ as $[X] \in L^p$, $[X] =\{Y \in \mathcal{L}^p : Y \overset{P a.e.}{=} X\}$.\\

We recall that a \underline{seminorm} $s$ is a function from a vector space $\mathcal{V}$ over the reals to ${\mathbb R}^+$, $s:\mathcal{V} \to {\mathbb R}^+$, such that $\forall a,b \in \mathcal{V}, \alpha \in {\mathbb R}$:
\begin{enumerate}
	\item $s(a+b) \leqslant s(a) + s(b)$\,\,\, (\textit{triangle inequality})
	\item $s(\alpha a) = |\alpha| a$\,\,\, (\textit{absolute homogeneity})
\end{enumerate}
Notice that from property $(2)$ non-negativity follows: $s(a) \geqslant 0$ and $s(\boldsymbol{0}) = 0$.\\
If furthermore the function $s$ \textit{separates points}, \textit{i.e.}
\begin{enumerate}
	\setcounter{enumi}{2}
	\item $s(a) = 0 \Rightarrow a = \boldsymbol{0}$
\end{enumerate}
then $s$ is a \underline{norm}, which we typically indicate with double bars, $||\cdot||$.

While $\mathcal{L}^p$ is only a seminormed space, $\|[X]\|_p := \|X\|_p=(\int_{\Omega}|X|^p d\mu)^{\frac{1}{p}} $ for any $X \in [X]$ \, is a norm in the quotient space $L^p = \mathcal{L}^p/\overset{P a.e.}{=}$. \\
By an abuse of notation we will say in the following $X \in L^p$ when we really mean $[X] \in L^p$. \\
It can finally be shown that  the normed space $L^p (\Omega,\mathcal{F},P))$, $p\geq 1$, is complete, hence it is a Banach space.

\subsubsection*{\underline{The Hilbert space of finite-variance random variables, $\mathcal{H} = L^2(\Omega,\mathcal{F},P)$}}

Suppose that in the above construction we take $p = 2$. Clearly, the set of measurable functions $X$ on $\Omega$ such that $(\int_{\Omega}|X|^2 dP)^{\frac{1}{2}} <\infty$ coincides with the set  $\{X : \int_{\Omega} X^2 dP <\infty \}$, which is the set of all square-integrable measurable functions on $\Omega$ or the set of all finite-variance random variables on $(\Omega,\mathcal{F},P)$ .\\
If we endow $L^2(\Omega,\mathcal{F},P)$ with the inner product $\langle X,Y\rangle = \int_{\Omega} X Y \, dP$, which is obviously compatible with the norm, we make it into an inner product space. Since, as noted above, it is a Banach space, it is also a Hilbert space.\\
We denote as $\mathcal{H}$,  $\mathcal{H} =  L^2(\Omega,\mathcal{F},P)$, this Hilbert space of random variables with finite variance, with $\langle X,Y \rangle = E(XY)= \int_{\Omega} X Y \, dP $.\\
Notice that $L^p(\Omega,\mathcal F, P)$ is not a Hilbert space for $p\neq2$: we will thus work strictly in $L^2(\Omega,\mathcal F, P)$.

\subsubsection*{\underline{Key facts about Hilbert spaces}}

We state in the following a few useful results, proofs can be found in any standard text on functional analysis.

\begin{theorem}[Canonical norm]
    Any inner product space $\mathcal{X}$ can be endowed with the norm
    \begin{equation*}
    	\|x \|:={\sqrt {\langle x ,x \rangle }}, \quad x \in \mathcal{X}.
    \end{equation*}
    This is called the \underline{inner product norm}, or \underline{canonical norm}, or the norm \underline{induced} by the inner product
\end{theorem}

\begin{definition}
    The \underline{distance} between two elements of an Hilbert space $x,y \in \mathcal X $ is $\mathrm{d}(x,y)=\|x-y\|$, where $\| \cdot \| $ is the canonical norm
\end{definition}

\begin{definition}
    $x,y \in \mathcal X $ are called \underline{orthogonal} if $\langle x ,y \rangle =0 $
\end{definition}

\begin{theorem}
    For every point $x$ in a Hilbert space $\mathcal{H}$ and every nonempty closed convex subset $C \subseteq H$ there exists a unique vector $y \in C$ for which $\lVert x-y\rVert$ is equal to $\mathrm{d}(x,C) :=\underset{c\in C}{\inf}\|x-c\| $
\end{theorem}

\begin{theorem}[Hilbert projection theorem]
    Let $\mathcal{H}_1 \subset \mathcal H$ be a subspace of a Hilbert space $\mathcal H$. Then, $\forall x \in \mathcal{H},\ \exists$ a unique vector $\mathcal{P}x\in \mathcal{H}_1$ such that $\forall z\in \mathcal{H}_1,\ \langle x-\mathcal{P}x, z \rangle=0$. $\mathcal{P}x$ is called the \underline{projection} of $x$ onto $\mathcal{H}_1$. In addition, $\mathcal{P}x$ is the element of minimum distance from $\mathcal{H}_1$ to $x$.
\end{theorem}

\begin{proposition}
    The operator $\mathcal{P}$ above has the following properties:
    \begin{itemize}
        \item $\mathcal{P}$ is a linear contraction
        \item $\mathcal{P}$ is idempotent
        \item $\mathcal{P}$ is a bounded and self-adjoint operator ($\forall \, x, y \in \mathcal{H} \,\,\, \left <\mathcal{P} x,y \right > = \left < x, \mathcal{P} y \right >$)
    \end{itemize}
\end{proposition}

\begin{proposition}
    If $C$ is a closed vector subspace of a Hilbert space, then it is a Hilbert subspace. In addition, every finite dimensional vector subspace of an Hilbert space is closed.
\end{proposition}

\subsection{The nature of mathematical models: definitions}

We have now the motivation and have available all the necessary tools for a formal treatment of what mathematical models for the applied sciences are. We refer to the above definition of $\mathcal{H}$.

\begin{definition}
    Given a random variable of interest or objective random variable, $\boldsymbol{X}^o$, possibly supposed to be related to some random variable $\boldsymbol{U}$, a \underline{\textbf{model}} of $\boldsymbol{X}^o$ is an operator $\boldsymbol{X}(\boldsymbol{\theta}, \boldsymbol{U}) :  \mathbb{R}^q \times \mathcal{H} \rightarrow \mathcal{H}$, whose arguments are some parameter $\boldsymbol{\theta} \in \mathbb{R}^q$ together with the predictor random variable(s) $\boldsymbol{U}$ (if any), and whose value is another random variable $\boldsymbol{X} = \boldsymbol{X}(\boldsymbol{\theta}, \boldsymbol{U})$.
\end{definition}

\begin{remark}
    That $\boldsymbol{X}$ is a random variable follows from the Doob-Dynkin lemma. If $\boldsymbol{X}(\boldsymbol{\theta}, \boldsymbol{U})$ is to be a good model, then $\boldsymbol{X}$ should be close to $\boldsymbol{X}^o$ in the metric of $\mathcal{H}$.
\end{remark}

\begin{definition}
    Having fixed the predictor(s) $\boldsymbol{U}$\,, we will typically consider the relationship between the parameter and the value of the operator, the \underline{\textbf{model map}} $\boldsymbol{X} = \boldsymbol{X}(\boldsymbol{\theta})$ \,.
\end{definition}

\begin{remark}
    Notice that the above is a slight abuse of notation: since we are typically interested in the relationship between the parameter $\boldsymbol{\theta}$ and the value $\boldsymbol{X}$ of the model $\boldsymbol{X}(\boldsymbol{\theta}, \boldsymbol{U})$, having for all intents and purposes fixed $\boldsymbol{U}$ we just consider $\boldsymbol{X}$ as a function of $\boldsymbol{\theta}$. To keep the two concepts distinct, however, we use the two different names "model" and "model map".
\end{remark}

For clarity, we distinguish between the model map defined above (having values in $\mathcal H$) and the corresponding model function, of the same functional form but having values in ${\mathbb R}^n$:

\begin{definition}
	A \underline{\textbf{model function}} is a map $\boldsymbol{x}(\boldsymbol{\theta}) :  \mathbb{R}^q \rightarrow {\mathbb R}^n$, whose argument is some parameter $\boldsymbol{\theta} \in \mathbb{R}^q$ (considering any possible predictors fixed) and whose value is a point in ${\mathbb R}^n$,  $\boldsymbol{x} : {\mathbb R}^q \to {\mathbb R}^n,\,\,\,  = \boldsymbol{x}: \boldsymbol{\theta} \mapsto \boldsymbol{x}(\boldsymbol{\theta})$.
\end{definition}

\begin{definition}
    Having fixed $\boldsymbol{U}$ and assuming sufficient regularity of the model, the set of values ${\mathcal{M}_{\boldsymbol{X}}} = \{ \boldsymbol{X}(\boldsymbol{\theta}), \boldsymbol{\theta} \in \boldsymbol{\Theta}\}$, indexed by the parameter, is the \underline{\textbf{model manifold}}, ${\mathcal{M}_{\boldsymbol{X}}} \subset \mathcal{H}$
\end{definition}

\subsection{Linear and nonlinear models}

If the hypothesized model is linear, \textit{e.g.} $     \boldsymbol{X}(\boldsymbol{\theta}) = \boldsymbol{U} \, \boldsymbol{\theta} $\,, then it is easy to see that
${\mathcal{M}_{\boldsymbol{X}}}$ is a vector space with elements in $\mathcal{H}$, hence a subspace of the Hilbert space $\mathcal H = L^2(\Omega,\mathcal{F},P)$.\\
As a subspace of a Hilbert space, ${\mathcal{M}_{\boldsymbol{X}}}$ in this case inherits its Hilbert space structure (norm, inner product), and a projection operator is guaranteed. Consequently, the point of projection $\hat{\boldsymbol{X}} = \mathcal P_{{\mathcal M}_{\boldsymbol{X}}}(\boldsymbol{X}^o)$ is the one of minimal distance from $\boldsymbol{X}^o$ (in the norm of $\mathcal{H}$), is unique and \,\, ($\boldsymbol{X}^o - \hat {\boldsymbol{X}}$)\,\, is orthogonal to all elements of ${\mathcal{M}_{\boldsymbol{X}}}$.\\

If on the other hand the model is nonlinear, \textit{e.g.} $     \boldsymbol{X}(\theta) = \sin( \boldsymbol{U} \boldsymbol{ \theta})$ \,, then ${\mathcal{M}_{\boldsymbol{X}}}$ is not in general convex, much less a subspace of $\mathcal{H}$. In our example in fact ${\mathcal{M}_{\boldsymbol{X}}} :=\{\boldsymbol{X}(\boldsymbol{\theta})\} $ is not convex since for $\boldsymbol{\theta}^1, \boldsymbol{\theta}^2 \in \boldsymbol{\Theta}, \boldsymbol{\theta}^1 \neq \boldsymbol{\theta}^2 $ we have $\boldsymbol{X}(\theta^1) + \boldsymbol{X}(\theta^2) = \sin(\boldsymbol{U} \boldsymbol{\theta}^1) + \sin(\boldsymbol{U} \boldsymbol{\theta}^2) \notin {\mathcal{M}_{\boldsymbol{X}}}$. Hence ${\mathcal{M}_{\boldsymbol{X}}}$ is not a vector space, not a subspace of the Hilbert space $\mathcal H =  L^2(\Omega,\mathcal{F}, P)$, no projection operator $\mathcal P_{{\mathcal M}_{\boldsymbol{X}}}$ can be defined and no unique ``best'' estimate $\hat {\boldsymbol{X}}$ can be obtained.\\

However, assume $\exists$ open $ \, \mathcal{U} \subset {\mathcal{M}_{\boldsymbol{X}}}$ and $ \hat{\boldsymbol{X}} = \boldsymbol{X}(\hat{\boldsymbol{\theta}})\in\mathcal{U}$ such that  $\forall \boldsymbol{X}\in\mathcal{U}$, $\boldsymbol{X} \neq \hat {\boldsymbol{X}}$ it is  $\mathrm{d}({\boldsymbol{X}}^o,\hat{\boldsymbol{X}})<\mathrm{d}(\boldsymbol{X}^o,\boldsymbol{X})$, \textit{i.e.} assume there exists an isolated point $\hat {\boldsymbol{X}} \in {\mathcal{M}_{\boldsymbol{X}}}$  locally closest to ${\boldsymbol{X}}^o$. Under this hypothesis, that a random variable $\hat {\boldsymbol{X}}$ exists, of minimal distance from $\boldsymbol{X}^o$ within a neighborhood (in the topology on $\mathcal{H}$ induced by the norm), then we could consider working on a local linear approximation to the manifold ${\mathcal{M}_{\boldsymbol{X}}}$ at $\hat {\boldsymbol{X}}$, that is we could consider constructing the tangent space to the manifold ${\mathcal{M}_{\boldsymbol{X}}}$ at that locally optimal point. To do so, we need to use the concept of the tangent space to a generic topological manifold at one of its points. Any good introductory text in differential geometry\cite{Tu2010, Lee2009, Lang1995} will provide the relevant definitions and proofs.

A clarification is needed here: when speaking of a ``tangent space'' to a manifold in a differential-geometric context what is usually meant is a suitable vector space (of differential operators, see below and Appendix \ref{sec:AffineTangentToM}), which we will call in the present instance \underline{direction space} and denote by $T_{\hat{\boldsymbol{X}}}{\mathcal{M}_{\boldsymbol{X}}}$. However, we are interested here in the affine ``space'' of random variables $\mathcal{T}_{\hat{\boldsymbol{X}}} = \hat{\boldsymbol{X}} + T_{\hat{\boldsymbol{X}}}{\mathcal{M}_{\boldsymbol{X}}}$\, locally approximating $\mathcal{M}_{\boldsymbol{X}}$ in a neighborhood of $\hat{\boldsymbol{X}}$: we will thus call $\mathcal{T}_{\hat{\boldsymbol{X}}}$ the \underline{tangent space}, even though, clearly, it does not have in general the structure of a vector space.

Now, the Hilbert projection theorem extends naturally to affine subspaces by translation, see Appendix \ref{sec:HilbertAffine} for details: here we will directly deal with the local approximation  $\mathcal{T}_{\hat{\boldsymbol{X}}}$ to ${\mathcal{M}_{\boldsymbol{X}}}$ at $\hat {\boldsymbol{X}}$, calling it, with a slight abuse of language, a ``subspace'' of $\mathcal{H}$, see Fig.\ref{fig:NonLinearHilbertReparametrized}.

We can build the direction space $T_{\hat {\boldsymbol{X}}}{\mathcal{M}_{\boldsymbol{X}}}$ by considering the set of velocities, \textit{i.e.} of differential operators, which map functions in $C^{\infty}({\mathcal{M}_{\boldsymbol{X}}})$ , $f:{\mathcal{M}_{\boldsymbol{X}}}\rightarrow\mathbb{R}$,\,  to their directional derivatives in the direction of curves $\gamma:\mathbb{R}\rightarrow {\mathcal{M}_{\boldsymbol{X}}}$, \textit{i.e.}:\\ $T_{\hat {\boldsymbol{X}}}{\mathcal{M}_{\boldsymbol{X}}} = \{\vartheta_{\gamma,\hat {\boldsymbol{X}}}$ such that $\, \hat {\boldsymbol{X}} =\gamma(\lambda)\in \mathcal{U}\subset {\mathcal{M}_{\boldsymbol{X}}}$ and $\forall f \in C^{\infty}({\mathcal{M}_{\boldsymbol{X}}}) $ it is $\, \vartheta_{\gamma,\hat {\boldsymbol{X}}} f = (f\circ\gamma)'|_{\lambda}\}$.\\ Notice that here the curve $\gamma$ is just one representative of a family of equivalent curves $\{\gamma^\bullet : \gamma^\bullet(\lambda^\bullet) = \gamma(\lambda) = \hat {\boldsymbol{X}}, (f\circ\gamma^\bullet)'|_{\lambda^\bullet} =(f\circ\gamma)'|_{\lambda}\}$.

It is now necessary both to establish a basis for $T_{\hat {\boldsymbol{X}}}{\mathcal{M}_{\boldsymbol{X}}}$ and to clarify the correspondence between the differential operators in $T_{\hat {\boldsymbol{X}}}{\mathcal{M}_{\boldsymbol{X}}}$ and the random variables that ought to belong to $\mathcal{T}_{\hat{\boldsymbol{X}}}$.

We can do this using the ``canonical'' chart from the manifold ${\mathcal{M}_{\boldsymbol{X}}}$ to $\mathbb{R}^q$ induced by the parametrization of the model. We can in general define a chart-induced basis of a tangent space: in the present case we will compute the basis for $T_{\hat {\boldsymbol{X}}}{\mathcal{M}_{\boldsymbol{X}}}$ induced by the \underline{canonical chart} $\boldsymbol{X}^{-1}$, similarly to the canonical chart map from $\mathbb{R}^n$ to $\mathbb{R}^q$ already introduced at line \ref{canonicalchartmap1} of page \pageref{canonicalchartmap1}: \\

\begin{center}
    \begin{tikzpicture}
        \begin{scope}[every node/.style={draw=white!60,thick,draw}]
            \node (A) at (0,0) {${\mathcal{M}_{\boldsymbol{X}}}$};
            \node (B) at (0,-2) {$\mathbb{R}^q$};
        \end{scope}

        \begin{scope}[>={Stealth[black]},
            every node/.style={},
            every edge/.style={draw=black}]
            \path [->] (A) edge [bend right=50] [left] node {$\boldsymbol{X}^{-1}$} (B);
            \path [->] (B) edge [bend right=50][right] node {$\boldsymbol{X}$} (A);
        \end{scope}
    \end{tikzpicture}

\end{center}

Clearly we must require sufficient regularity in the model map $\boldsymbol{X}$ to do this, such as being continuously differentiable with a continuously differentiable inverse.
Notice that $\boldsymbol{X}^{-1}$ is a chart, indeed a global chart, mapping the whole of ${\mathcal{M}_{\boldsymbol{X}}}$ to $\mathbb{R}^q$, since ${\mathcal{M}_{\boldsymbol{X}}}$ is generated by $\boldsymbol{X}(\boldsymbol{\Theta})$, with $\boldsymbol{\Theta} \in \mathbb{R}^q$, because by definition for every admissible $\boldsymbol{\theta}$ there is a point $\boldsymbol{X}(\boldsymbol{ \theta}) \in {\mathcal{M}_{\boldsymbol{X}}}$. \\
We can also say, therefore that:

\begin{lemma}
    The inverse of a model map is a chart map.

\end{lemma}

Summarizing what is more precisely detailed in Appendix \ref{sec:AffineTangentToM}, we can thus use as basis of the tangent space $T_{\hat{\boldsymbol{X}}}{{\mathcal{M}_{\boldsymbol{X}}}}$ the basis induced by the chart map $(\boldsymbol{X}^{-1})$:

\begin{equation*}
    \{(\cdot \circ ((\boldsymbol{X}^{-1})^{-1})^j)'|_{\hat{\boldsymbol{X}}}\}_{j=1,...,q}=\{(\cdot \circ \boldsymbol{X}^j)'|_{\hat{\boldsymbol{X}}}\}_{j=1,...,q} \,\,.
\end{equation*}

Denoting  $\dot{\boldsymbol{X}} := \boldsymbol{X}'$ and writing $(\boldsymbol{X}^j)'=(\dot{\boldsymbol{X}})^j=\dot{\boldsymbol{X}}^j$ \,,\,\, the tangent space is the span  $T_{\hat{\boldsymbol{X}}}{{\mathcal{M}_{\boldsymbol{X}}}} = {\mathcal{S}}\{ \dot{\boldsymbol{X}}^j\}_{j=1,...,q}$ \,\, WRT the canonical model chart map $(\boldsymbol{X}^{-1})$. This means that for any ${\boldsymbol{\eta}}\in\mathbb{R}^q$, $ \dot{\boldsymbol{X}} \cdot {\boldsymbol{\eta}}\in T_{\hat{\boldsymbol{X}}}{\mathcal{M}_{\boldsymbol{X}}}$ and that $\mathcal{T}_{\hat{{\boldsymbol{X}}}} =   \hat{\boldsymbol{X}} + T_{\hat{\boldsymbol{X}}}{\mathcal{M}_{\boldsymbol{X}}} = \hat{\boldsymbol{X}} + \{ \dot{\boldsymbol{X}} \cdot {\boldsymbol{\eta}}\,,\,\,{\boldsymbol{\eta}}\in\mathbb{R}^q \}\subset \mathcal{H}$ is an affine space of  random variables tangent to the model manifold ${\mathcal{M}_{\boldsymbol{X}}}$ at $\hat{\boldsymbol{X}}$.\\

An affine subspace $\mathcal{T}_{\hat{{\boldsymbol{X}}}}$ of random variables has thus been constructed using the model-induced canonical chart map $\boldsymbol{X}^{-1}$. The Hilbert projection theorem applies and the random variable $\boldsymbol{X}^o$ has a unique projection onto $\hat{\boldsymbol{X}} \in \mathcal{T}_{\hat{{\boldsymbol{X}}}}$ , of global minimal distance from $\mathcal{T}_{\hat{{\boldsymbol{X}}}}$ in the metric induced by the norm of $\mathcal{H}$, such that $\forall \, \boldsymbol{X} \in \mathcal{T}_{\hat{{\boldsymbol{X}}}}$ we have $(\boldsymbol{X}^o - \hat{\boldsymbol{X}}) \, \perp (\boldsymbol{X} - \hat{\boldsymbol{X}})$.

\subsection{\underline{Mapping $\mathcal{H}$ to $\mathbb{R}^n$}}

It is now possible to formalize the intuitively obvious relationship between the theoretical construction of the abstract model in the Hilbert space of finite-variance random variables $\mathcal{H}$ and the actual computational, numerical analysis of the results carried out with statistical techniques in case space $\mathbb{R}^n$.

Recall that a random variable is defined as a measurable function from   $(\Omega,\mathcal{F},P)$ to $\mathbb{R}^n$, mapping outcomes to (vectors of) reals, say

\begin{equation*}
	{\boldsymbol{Y}} \in \mathcal H \quad  \Rightarrow  \quad {\boldsymbol{Y}}:\Omega \rightarrow \mathbb{R}^n\,\,,\,\,\,{\boldsymbol{Y}}: \omega \mapsto r\,.
\end{equation*}

In fact, the act of sampling some variable on some experimental unit (a person from a population of interest, the roll of a die, ...) can be seen as a map from the product space $\Omega \times {\mathcal{H}}$ to $\mathbb{R}^n$, of which ${\boldsymbol{Y}}:\Omega \rightarrow \mathbb{R}^n$ is the marginal map obtained by fixing ${\boldsymbol{Y}} \in {\mathcal{H}}$. In the same way we may consider the other marginal map, obtained by fixing $\omega \in \Omega$: we can thus consider the outcome  as a map from $\mathcal{H}$ to $\mathbb{R}^n$:

\begin{equation*}
	\omega \in \Omega \quad  \Rightarrow  \quad \omega: \mathcal H \rightarrow \mathbb{R}^n\,\,,\,\,\,\omega:{\boldsymbol{Y}} \mapsto  r   \,.
\end{equation*}

In other words, we may define

\begin{equation}
    \omega({\boldsymbol{Y}}) := {\boldsymbol{Y}}(\omega) = \boldsymbol{r} \in \mathbb{R}^n
\end{equation}

Experimental sampling is the concrete act of picking an outcome $\omega$ from the sample space $\Omega$. This act fixes $\omega$ and thus determines the specific realizations $\boldsymbol{x}^o, \, \boldsymbol{y}^o, \dots$ of any desired random variables $\boldsymbol{X}^o, \boldsymbol{Y}^o, \dots$. As a consequence and within the context of our modeling discussion, the sampling act also determines the realization $\hat{\boldsymbol{x}}$ of $\hat{\boldsymbol{X}}$, and so on. All random variables considered in the theoretical construction of the model are thus naturally mapped by $\omega$ to their respective numerical realizations. We see then that the event $\omega$ can be interpreted as the necessary map from $\mathcal H$ to $\mathbb R^n$, see Fig.\ref{fig:SamplingAsMap}.

\begin{figure}
    \includegraphics[height=7cm]{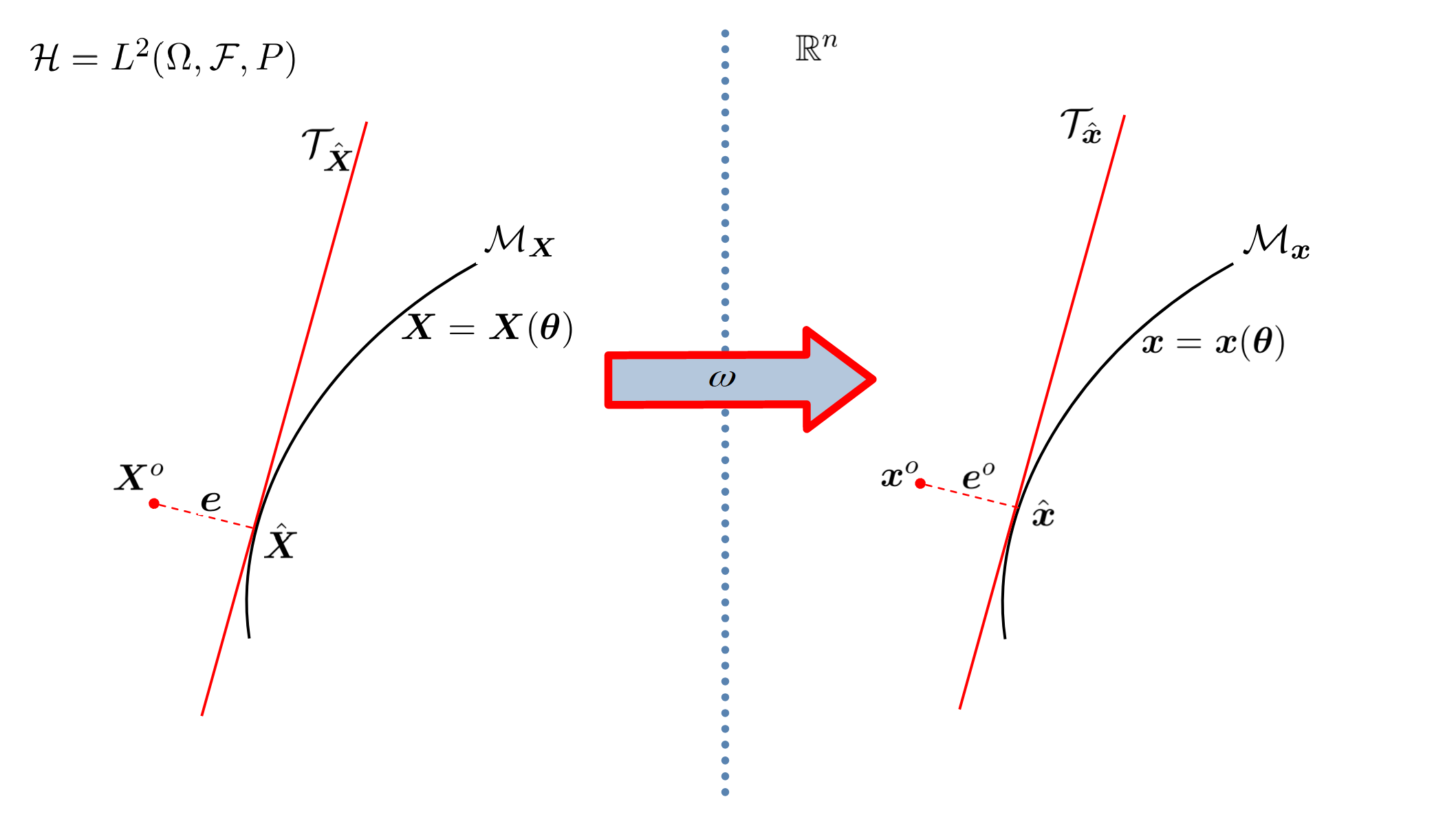}
    \caption{Sampling as a map from $\mathcal{H}$ to $\mathbb{R}^n$.}
    \label{fig:SamplingAsMap}
\end{figure}

\subsection{Summary}

We recap at this point the different elements of the twin structures (in $\mathcal{H}$ and in $\mathbb{R}^n$) that we have described. While the use of many of these elements corresponds to standard statistical practice, the following list has the purpose of systematizing the notation to dispel possible sources of confusion, as well as of clarifying the correspondence of objects in the abstract modeling and the concrete computational settings.

\begin{itemize}

    \item We are interested in predicting the values of some objective random variable $\boldsymbol X^o \in \mathcal{H}$;\\

    \item to do so we formulate a conceptual model as the operator $\boldsymbol X=\boldsymbol X(\boldsymbol \theta, \boldsymbol U)$, yielding a random variable $\boldsymbol X \in \mathcal{H}$ for any given value of $\boldsymbol \theta \in \mathbb{R}^q$. We hope to find some optimal value of the parameter $\boldsymbol  \theta$ such that the random variable $\boldsymbol X$  is close to the objective random variable $\boldsymbol X^o$, this would give us a ``good'' model.\\

    \item $\boldsymbol x^o$ are measurements of some quantity of interest, $\boldsymbol x^o \in \mathbb{R}^n$. We approximate $\boldsymbol x^o$ by means of some other quantity $\boldsymbol x =\boldsymbol x(\boldsymbol \theta, \boldsymbol u) = \boldsymbol x(\boldsymbol \theta)$ computed as a function of some parameter $\boldsymbol \theta \in \Theta \subset \mathbb{R}^q$ (and possibly of some predictor values $\boldsymbol u$): the fundamental idea is that we obtain a forecast $\boldsymbol x$ for each arbitrary value of $\boldsymbol \theta$.\\
    Notice that we distinguish between a model, \textit{i.e.} a model operator $\boldsymbol{X}(\boldsymbol{\theta}, \boldsymbol{U}) : \mathbb{R}^q \times \mathcal{H} \rightarrow \mathcal{H}$, and a model function $\boldsymbol{x}(\boldsymbol{\theta}, \boldsymbol{u}) : \mathbb{R}^q \times \mathbb{R}^n \rightarrow \mathbb{R}^n$, even though their functional form is the same.\\ 
    For convenience, disregarding any predictor variable(s) we also define the model map $\boldsymbol{X}(\boldsymbol{\theta}): \mathbb{R}^q \rightarrow \mathcal{H}$, which is the inverse of the model-induced (global) canonical chart $\boldsymbol{X}^{-1}: \mathcal{H} \rightarrow \mathbb{R}^q$.\\
    
    \item So $\boldsymbol \theta$ is some arbitrary value of the parameter; we believe $\boldsymbol \theta^*$ to be the "true" but unknown value of the parameter, to which therefore corresponds a true but unknown $\boldsymbol x^* = \boldsymbol x(\boldsymbol \theta^*)$.\\

    \item When making an experiment, we observe the value $\boldsymbol x^o$, which in general cannot be produced exactly by the function $\boldsymbol x(\boldsymbol \theta)$ for any value of $\boldsymbol \theta$; since we have assumed that the relationship is true (with $\boldsymbol x^* = \boldsymbol x(\boldsymbol \theta^*)$), we deduce therefore that some (observation) error has been made in measuring $\boldsymbol x^o$.\\

    \item We denote this error by $\boldsymbol \varepsilon^o = \boldsymbol x^o - \boldsymbol x^*$ (even though $\boldsymbol \varepsilon^o$ is not actually observed): we thus understand the observation error that one might make as some random variable $\boldsymbol \varepsilon$, of which $\boldsymbol \varepsilon^o$ is the specific realization obtained in our experiment.\\

    \item We indicate with $\boldsymbol X^o$ the random variable $\boldsymbol X^o = \boldsymbol X(\boldsymbol \theta^*) + \boldsymbol \varepsilon$; $\boldsymbol x^o$ is the observed realization of $\boldsymbol X^o$, $\boldsymbol x^o = \boldsymbol x(\boldsymbol \theta^*)+\boldsymbol \varepsilon^o$, where both $\boldsymbol \theta^*$ and $\boldsymbol \varepsilon^o$ are however unknown.\\

    \item The predictors $\boldsymbol u$, which we observe, may be regarded as the sample realization of some (design) random variable $\boldsymbol U$ \,\,(possibly having zero variance, \textit{e.g.} where fixed by design), hence $\boldsymbol X = \boldsymbol X(\boldsymbol \theta) = \boldsymbol X(\boldsymbol \theta, \boldsymbol U)$, for any given $\boldsymbol \theta$, is also a random variable.\\

    \item In order to assess the likely value of the unknown parameter $\boldsymbol \theta^*$ we introduce an estimator $\hat{\boldsymbol \theta}$, which is a function of the random variable $\boldsymbol X^o$ and hence a random variable itself, $\hat{\boldsymbol \theta} = \hat{\boldsymbol \theta}(\boldsymbol  X^o)$. The statistic $\hat{\boldsymbol \theta}^o = \hat{\boldsymbol \theta}(\boldsymbol  x^o)$ is the observed realization of $\hat{\boldsymbol \theta}$. \\

    \item $\hat{\boldsymbol X} = \boldsymbol X(\hat{\boldsymbol \theta})  = \boldsymbol X(\hat{\boldsymbol \theta}(\boldsymbol X^o)) $ is the estimator of $\boldsymbol X^o$ assuming the model $\boldsymbol X=\boldsymbol X(\boldsymbol \theta) =\boldsymbol X(\boldsymbol \theta, \boldsymbol U)$ to be true. Being a function of a random variable, it is of course a random variable itself; its observed realization is $\hat{\boldsymbol x}^o = \boldsymbol x(\hat{\boldsymbol \theta}^o)$. \\
    
    \item $\boldsymbol e$ is the random variable $\boldsymbol e = \boldsymbol X^o - \boldsymbol X = \boldsymbol X^o - \boldsymbol X(\boldsymbol \theta)$ corresponding to some arbitrary value of $\boldsymbol \theta$; $\hat{\boldsymbol e}$ is the random variable $\hat{\boldsymbol e} = \boldsymbol X^o - \hat{\boldsymbol X} = \boldsymbol X^o - \boldsymbol X(\hat{\boldsymbol \theta})$ corresponding to the estimator $\hat{\boldsymbol \theta}$; $\boldsymbol e^o$ and $\hat{\boldsymbol e}^o$ are observed realizations of the respective random variables corresponding to arbitrary $\boldsymbol \theta$ or estimated $\hat{\boldsymbol \theta}^o$.\\

    \item $\boldsymbol \xi$, the ``\underline{flaw}" (term introduced above at line \lineref{ref:theflaw}, page \pageref{ref:theflaw}), is the random variable given by $\boldsymbol \xi (\boldsymbol \theta)  = \boldsymbol X - \boldsymbol X^* = \boldsymbol X(\boldsymbol \theta, \boldsymbol U) - \boldsymbol X(\boldsymbol \theta^*, \boldsymbol U)$; we call the random variable $\hat{\boldsymbol \xi} =  \boldsymbol \xi (\hat{\boldsymbol \theta})  = \hat{\boldsymbol X} - \boldsymbol X^* $  the flaw estimator; the flaw corresponding to the observed sample is of course $\hat{\boldsymbol \xi}^o =  \boldsymbol \xi (\hat{\boldsymbol \theta}^o)  = \hat{\boldsymbol x}^o - \boldsymbol x^* $, it evidently cannot be observed.\\

    \item ${\mathcal{M}_{\boldsymbol{X}}} = \{ \boldsymbol{X} = \boldsymbol{X}(\boldsymbol{\theta}, \boldsymbol{U}), \boldsymbol{\theta} \in \Theta\}$ is the model manifold, the set of all possible random variables $\boldsymbol{X}$ generated by the model operator. Having observed or set $\boldsymbol{u}$, the prediction surface  ${\mathcal{M}_{\boldsymbol{x}}} = \{ \boldsymbol{x} = \boldsymbol{x} (\boldsymbol{\theta},\boldsymbol{u}), \boldsymbol{\theta} \in \Theta \}$ is the set of all possible values $\boldsymbol{x} \in \mathbb{R}^n$ which can be generated by the model function $\boldsymbol{x} (\boldsymbol{\theta},\boldsymbol{u})$.\\
    
    \item $T_{\hat{\boldsymbol{X}}}{{\mathcal{M}_{\boldsymbol{X}}}}$ is the direction space (the tangent space in differential-geometric parlance), the vector space of differential operators defined at $\hat{\boldsymbol{X}}$ with respect to the canonical chart map $\boldsymbol{X}^{-1}$, which is the inverse of the model map $\boldsymbol{X}({\boldsymbol{\theta}})$.\\
    
    \item $\mathcal{T}_{\hat{{\boldsymbol{X}}}} =   \hat{\boldsymbol{X}} + T_{\hat{\boldsymbol{X}}}{\mathcal{M}_{\boldsymbol{X}}} = \hat{\boldsymbol{X}} + \{ \dot{\boldsymbol{X}} \cdot {\boldsymbol{\eta}}\,,\,\,{\boldsymbol{\eta}}\in\mathbb{R}^q \}\subset \mathcal{H}$ is the affine space of  random variables tangent to the model manifold ${\mathcal{M}_{\boldsymbol{X}}}$ at $\hat{\boldsymbol{X}}$.\\

    \item Notice that $\boldsymbol{x}^* = \boldsymbol{x}(\boldsymbol{\theta}^*,\boldsymbol{u})$ for fixed $\boldsymbol{u}$ is not the realization of a random variable, but it would be if $\boldsymbol{u}$ were a realization of the random variable $\boldsymbol{U}$. In other words, if the design $\boldsymbol{U}$ is fixed, then $\boldsymbol{u} \equiv \boldsymbol{U}$ deterministic and $\boldsymbol{x}^*$ is also fixed (but unknown). More generally we could consider $\boldsymbol{x}^*$ to be the realization of the random variable $\boldsymbol{X}^* = \boldsymbol{x}(\boldsymbol{\theta}^*,\boldsymbol{U})$, possibly having zero variance.
    
	\item Sometimes we will need to generate samples of artificial realizations of $\boldsymbol X$, which we would denote by $\tilde{\boldsymbol x}^r, r=1, 2, \dots$.\\

\end{itemize}


\section{Discussion and Conclusions}

\noindent Recent literature addresses similar problems (in particular the characterization of the projection) in the case of conditional expectation \cite{Bobrowski2013} and in the case of information geometry \cite{BarndorffNielsen1986, Fletcher2020, Brigo2011,Nielsen2020, Amari2009}.

\subsection{Relationship with conditional estimation}

As already mentioned in the Introduction, much work has been done on the Hilbert space of all finite-variance random variables on a probability space $(\Omega,\mathcal{F},P)$, \textit{i.e.} on all $\mathcal{F}$-measurable functions on $\Omega$. Estimation of a random variable $\boldsymbol{X}^o \in \mathcal{H}$ conditional on a sub-$\sigma$-algebra $\mathcal{F}_\mathcal{G} \subset \mathcal{F}$, \textit{i.e.} assuming a lesser amount of information than that provided by the full $\sigma$-algebra $\mathcal{F}$, is shown to consist in the projection of the original random variable $\boldsymbol{X}^o$ onto the subspace $\mathcal{G} \subset \mathcal{H}$ of $\mathcal{F_G}$-measurable random variables.
In this case we do have a (proper) subspace, the Hilbert projection theorem holds and the projection of $\boldsymbol{X}^o \in \mathcal{H}$ onto $\mathcal{G}$ is the conditional estimate $\hat{\boldsymbol{X}}$ (conditional on $\mathcal{F_G}$).\\

Notice however that the situation described in the present work is different: even in the case of linearity, the subspace of interest is generated by a functional form of the model, linear with respect to the parameter values, not by the restriction of information on the outcome characterized by an appropriate sub-$\sigma$-algebra. The difference between the two situations is obvious when considering a nonlinear functional form of the model.

\subsection{Relationship with Information Geometry}

The well-developed field of Information Geometry deals with spaces of distributions, wherein certain (parametrized) families are embedded manifolds. Interesting properties can be derived for some such families (\textit{e.g.} the exponential family generates a ``flat'' manifold in the appropriate Riemannian metric tensor because, under the Fisher information metric and exponential connection, the parameter space forms an affine, curvature-free geometry).\\

Interesting facts emerge from the study of such spaces of distributions. For example, denote by $\mathcal{P}$ the space of probability distributions, and denote by $\mathcal{G} \subset \mathcal{P}$ the manifold representing a family of probability distributions indexed by some parameter $\boldsymbol{\theta} \in \boldsymbol{\Theta} \subset \mathbb{R}^q$, $G(\boldsymbol{\theta}) \in \mathcal{G}$, which family we wish to use as a possible ``model'' for our data-generating process. We assume that the data are actually generated by an unknown distribution $P \in \mathcal{P}$, with, in general, $P \notin \mathcal{G}$. If we observe data $\boldsymbol{x} = \{ x_1, ..., x_n\}$ we can build an empirical distribution $\tilde{P}$ out of them. We could then hypothesize that the ``best'' value of $\boldsymbol{\theta}$ is the value $\hat{\boldsymbol{\theta}}$ that minimizes some kind of distance or divergence from $\tilde{P}$ to $G(\boldsymbol{\theta})$: in the case the Kullback-Leibler divergence $D_{KL}$ of  $G(\boldsymbol{\theta})$ from $\tilde{P}$ is considered, we obtain the Maximum Likelihood Estimator $\hat {\boldsymbol{\theta}}_{ML}$, and if furthermore the family ${\mathcal{G}}$ is flat (\textit{e.g.} ${\mathcal{G}}$ is the exponential family) then estimating by MLE corresponds to projecting $\tilde{P}$ onto ${\mathcal{G}}$.\\

Now, there is a (biunivocal $P$-almost-everywhere) relationship between random variables on $(\Omega,\mathcal{F},P)$ and their distributions, in the sense that each random variable induces a distribution and conversely we can define a random variable from a distribution, two random variables with the same distribution being equal except on $P$-null sets. In other words, when we model a physical phenomenon we are working in $\mathcal{H}$, and when we work on the properties of estimators we use results established in $\mathcal{P}$, but there is an almost-everywhere 1-1 correspondence between the two.

Specifically, in information geometry a set of probability distributions indexed by a finite number of parameters is shown to be a (smooth) manifold and may be equipped with a Riemannian structure and with conjugate connections, which make the manifold `flat', thus generalizing the concept of Euclidean Space \cite{Nielsen2020,BarndorffNielsen1986}. This flat structure can be obtained in different ways, some of which are discussed in the papers of Amari et al. and Nielsen \cite{Amari2009, Nielsen2020, BarndorffNielsen1986,Amari2010}. One way of doing this is through Bregman divergence, which induces a structure on a manifold ${\mathcal{G}}$ that generalizes the Pythagorean Theorem and consequently guarantees the existence and uniqueness of the projection of a point $\tilde P$ onto the manifold ${\mathcal{G}}$ \cite{Amari2009, Nielsen2020}. A similar approach to Amari (also formulating a projection theorem) is proposed by Brigo \cite{Brigo2011} using a direct $\mathcal{L}^2$ metric.\\

It can therefore be seen that, as stated above, the theory of conditional expectation as a projection onto an appropriate subspace of the Hilbert space $\mathcal{H}$ of random variables is not related to the present discussion, because the underlying manifold is different. There is instead a clear connection with information geometry: to the best of the author's knowledge, however, the identification of the underlying manifold as determined by the mathematical model of the experiment, as well as the identification of the map from the space $\mathcal{H}$ (where the modeling is abstractly done) to the space $\mathbb{R}^n$ (where the computations are carried out) have not yet been formalized.\\

\section{Conclusions}

A mathematical model can properly be defined as an operator yielding a random variable $\boldsymbol{X}$ approximating the objective random variable $\boldsymbol{X}^o$ in $\mathcal{H}$. The outcome $\omega$ maps $\mathcal{H}$ onto case space $\mathbb{R}^n$. Estimation for linear models in case space consists in projecting the observed value $\boldsymbol{x}^o$ onto the prediction surface ${\mathcal{M}_{\boldsymbol{x}}}$, and corresponds via $\omega$ to projecting $\boldsymbol{X}^o$ onto the model manifold ${\mathcal{M}_{\boldsymbol{X}}}$. For nonlinear models a local optimum $\hat {\boldsymbol{x}} \in {\mathcal{M}_{\boldsymbol{x}}}$ can be found numerically, and is the image under $\omega$ of a locally optimal random variable $\hat{\boldsymbol{X}} \in {\mathcal{M}_{\boldsymbol{X}}}$. Asymptotic confidence regions can be computed on ${\mathcal{M}_{\boldsymbol{x}}}$ for linear models. For nonlinear models approximate asymptotic confidence regions can be computed on the affine tangent space $\mathcal{T}_{\hat {\boldsymbol{x}}}$ (constructed via computation of the Jacobian at the local optimum), which corresponds to the theoretical tangent space $\mathcal{T}_{\hat {\boldsymbol{X}}} \subset \mathcal{H}$.

\section{Appendices}

\subsection{Extension of Hilbert's theorem to affine spaces}
\label{sec:HilbertAffine}

The classical Hilbert projection theorem states that if $\mathcal{H}$ is a Hilbert space and $\mathcal{V} \subset \mathcal{H}$ is a nonempty closed convex subset, then for every $x \in \mathcal{H}$, there exists a unique point $p \in \mathcal{V}$ such that
\begin{equation}
    \|x - p\| = \inf_{y \in \mathcal{V}} \|x - y\|.
\end{equation}
In particular, when $\mathcal{V}$ is a closed linear subspace, the projection $p$ satisfies
\begin{equation*}
    x - p \perp \mathcal{V} .
\end{equation*}

Now, let $\mathcal{H}$ be a Hilbert space, and let ${\mathcal{A}}$ be an affine subspace, meaning there exists a closed linear subspace $\mathcal{V}$ and a fixed point $a \in \mathcal{H}$ such that
\begin{equation}
    {\mathcal{A}} = a + \mathcal{V} = \{ a + v : v \in \mathcal{V} \}.
\end{equation}

\begin{theorem}
    Let $\mathcal{H}$ be a Hilbert space, $\mathcal{V} \subset \mathcal{H}$ \,\,\,a nonempty closed subspace and ${\mathcal{A}} = a + \mathcal{V}$ a nonempty closed affine subspace. Then for every $x \in \mathcal{H}$, there exists a unique $p \in {\mathcal{A}}$ such that
    \begin{equation*}
        \|x - p\| = \inf_{y \in {\mathcal{A}}} \|x - y\|.
    \end{equation*}
    Moreover,
    \begin{equation*}
        x - p \perp \mathcal{V}.
    \end{equation*}
\end{theorem}

\begin{proof}
    Since ${\mathcal{A}}$ is an affine subspace, we can write ${\mathcal{A}} = a + \mathcal{V}$. Define the translation $x' = x - a$, so that our problem reduces to projecting $x'$ onto the linear subspace $\mathcal{V}$. By the Hilbert projection theorem for linear subspaces, there exists a unique $v_0 \in \mathcal{V}$ such that
    \begin{equation}
        \|x' - v_0\| = \inf_{v \in \mathcal{V}} \|x' - v\|.
    \end{equation}
    This implies that $x' - v_0$ is orthogonal to $\mathcal{V}$, i.e.,
    \begin{equation}
        x' - v_0 \perp \mathcal{V}.
    \end{equation}
    Setting $p = a + v_0 \in {\mathcal{A}}$, we obtain
    \begin{equation}
        x - p = (x' - v_0) \perp \mathcal{V}.
    \end{equation}
    Thus, $p$ is the unique closest point in ${\mathcal{A}}$ to $x$, and the orthogonality condition still holds relative to the associated direction subspace $\mathcal{V}$.
\end{proof}

\subsection{Tangent affine subspace to a nonlinear manifold in $\mathbb{R}^n$}
\label{sec:AffineTangentToS}

Consider a point $\boldsymbol{x}^o \in \mathbb{R}^n$ and suppose the relationship $\boldsymbol{x} = \boldsymbol{x}(\boldsymbol{\theta})$ is nonlinear, giving rise to a possibly non-convex prediction surface ${\mathcal{M}_{\boldsymbol{x}}}$. Suppose further that an isolated point $\hat{\boldsymbol{x}} \in {\mathcal{M}_{\boldsymbol{x}}}, \,\, \hat {\boldsymbol{x}} = {\boldsymbol{x}}(\hat{\boldsymbol{\theta}})$ exists of locally minimal distance from ${\boldsymbol{x}}^o$, \textit{i.e.} that there exists a neighborhood \,$\mathcal{U} \subset {\mathcal{M}_{\boldsymbol{x}}}\,,\,\,\, \hat{\boldsymbol{x}} \in \mathcal{U}$ \,such that $\| {\boldsymbol{x}}^o - \hat {\boldsymbol{x}}\| < \| {\boldsymbol{x}}^o - {\boldsymbol{x}}\| \,\,\,\forall {\boldsymbol{x}} \neq \hat {\boldsymbol{x}}$, see Figure \ref{fig:NonLinearCaseSpaceReparametrizedRn}.

Reparametrize the relationship, setting:

\begin{equation*}
	{\boldsymbol{\zeta}} := {\boldsymbol{\theta}} - \hat{{\boldsymbol{\theta}}} \quad, \quad \quad {\boldsymbol{\theta}} = {\boldsymbol{\zeta}} + \hat{{\boldsymbol{\theta}}}
\end{equation*}

and

\begin{equation*}
	{\boldsymbol{y}} \,=\, \boldsymbol{y}({\boldsymbol{\zeta}}) \,:=\, {\boldsymbol{x}} ({\boldsymbol{\zeta}} + \hat{{\boldsymbol{\theta}}}) -  \hat {\boldsymbol{x}} \,=\, {\boldsymbol{x}} ({\boldsymbol{\theta}}) -  \hat {\boldsymbol{x}} \,=\, {\boldsymbol{x}} ({\boldsymbol{\theta}}) -  {\boldsymbol{x}}(\hat{{\boldsymbol{\theta}}}) \,,
\end{equation*}

with

\begin{equation*}
	\mathcal{Y} = \{{\boldsymbol{y}}({\boldsymbol{\zeta}}), {\boldsymbol{\zeta}} \in {\boldsymbol{\Theta}} - \hat {\boldsymbol{\theta}}\} \,.
\end{equation*}

It follows that

\begin{equation*}
	\hat {\boldsymbol{y}} = {\boldsymbol{y}}(0)  = {{\boldsymbol{0}}},  \quad \quad {\boldsymbol{y}}^o  = {\boldsymbol{x}}^o - {\boldsymbol{x}}(\hat{{\boldsymbol{\theta}}}) \,,
	\quad \quad {\boldsymbol{y}}^* = {\boldsymbol{y}}({\boldsymbol{\zeta}}^*) = {\boldsymbol{y}}({\boldsymbol{\theta}}^* - \hat{\boldsymbol{\theta}})
\end{equation*}

where for clarity $0$ indicates the zero vector in $\mathbb{R}^q$ and ${\boldsymbol{0}}$ indicates the zero vector in $\mathbb{R}^n$.

Notice that

\begin{align*}
    \frac{\partial {\boldsymbol{y}}({\boldsymbol{\zeta}})}{\partial {\boldsymbol{\zeta}} } & =  \frac{\partial \left [{\boldsymbol{x}}({\boldsymbol{\theta}}) - \hat {\boldsymbol{x}} \right ]}{\partial {\boldsymbol{\zeta}} }\\
    & = \frac{\partial {\boldsymbol{x}}({\boldsymbol{\theta}})}{\partial {\boldsymbol{\zeta}} } \\
    & = \frac{\partial {\boldsymbol{x}}({\boldsymbol{\theta}})}{\partial {\boldsymbol{\theta}} }
    \frac{\partial {\boldsymbol{\theta}}}{\partial {\boldsymbol{\zeta}} } \\
    & = \frac{\partial {\boldsymbol{x}}({\boldsymbol{\theta}})}{\partial {\boldsymbol{\theta}} } \,,
\end{align*}

and in particular

\begin{align*}
    \left . \frac{\partial {\boldsymbol{y}}({\boldsymbol{\zeta}})}{\partial {\boldsymbol{\zeta}} } \right |_0 & =  \left . \frac{\partial \left [{\boldsymbol{x}}({\boldsymbol{\theta}}) - \hat{\boldsymbol{x}} \right ]}{\partial {\boldsymbol{\zeta}} } \right |_0\\
    & = \left . \frac{\partial {\boldsymbol{x}}({\boldsymbol{{\boldsymbol{\theta}}}})}{\partial {\boldsymbol{\zeta}} }  \right |_0\\
    & = \left . \frac{\partial {\boldsymbol{x}}({\boldsymbol{\theta}})}{\partial {\boldsymbol{\theta}} } \right |_{\hat {\boldsymbol{\theta}}}
    \left . \frac{\partial {\boldsymbol{\theta}}}{\partial {\boldsymbol{\zeta}} }  \right  |_0\\
    & = \left . \frac{\partial {\boldsymbol{x}}({\boldsymbol{\theta}})}{\partial {\boldsymbol{\theta}} } \right |_{\hat {\boldsymbol{\theta}}} \,, \quad \text{or}\\
    \dot{\boldsymbol{y}}(0) &= \dot{\boldsymbol{x}}(\hat{\boldsymbol{\theta}})
\end{align*}

We can thus construct the affine subspace

\begin{equation*}
	\tilde {\boldsymbol{x}} = \tilde {\boldsymbol{x}}({\boldsymbol{\eta}}) = \hat {\boldsymbol{x}} + \dot{{\boldsymbol{x}}} (\hat{{\boldsymbol{\theta}}}) \,\, {\boldsymbol{\eta}} = \hat {\boldsymbol{x}} +    \left .     \frac{\partial     {\boldsymbol{x}}({\boldsymbol{\theta}})}{\partial     {\boldsymbol{\theta}}} \right | _{\hat{{\boldsymbol{\theta}}}} \,\, {\boldsymbol{\eta}} \quad \quad {\boldsymbol{\eta}} \in \mathbb{R}^q
\end{equation*}

and the corresponding linear subspace

\begin{equation*}
	\tilde {\boldsymbol{y}} = \tilde {\boldsymbol{y}}({\boldsymbol{\eta}})
	=\dot{{\boldsymbol{y}}}(0) \,\,  {\boldsymbol{\eta}}
	= \left . \frac{\partial {\boldsymbol{y}}({\boldsymbol{\zeta}})}{\partial {\boldsymbol{\zeta}}} \right |_0 \,\,{\boldsymbol{\eta}} \,\, \quad \quad {\boldsymbol{\eta}} \in \mathbb{R}^q\,.
\end{equation*}

By construction, an isolated minimum of $\| {\boldsymbol{y}}^o - {\boldsymbol{y}} \|$ at $\hat {\boldsymbol{y}} = {\boldsymbol{y}}(0) = {\boldsymbol{0}}$ exists and we can define a (locally unique) projection operator $\mathcal P_{T_{\hat {\boldsymbol{y}}}{\mathcal{Y}}} : \mathbb R ^n \rightarrow T_{\hat {\boldsymbol{y}}}{\mathcal{Y}}$,  \quad $ \mathcal P_{T_{\hat {\boldsymbol{y}}}{\mathcal{Y}}} = \dot {\boldsymbol{y}} (\dot {\boldsymbol{y}} ^ T \dot{\boldsymbol{y}})^{-1} \dot {\boldsymbol{y}}^\intercal $ , $\mathcal P_{T_{\hat {\boldsymbol{y}}}{\mathcal{Y}}} {\boldsymbol{y}}^o = {\boldsymbol{0}}$, see Figure \ref{fig:NonLinearCaseSpaceReparametrizedRn}.

 It is then immediate to refer the statistical computations done on the linear subspace tangent to the reparametrized prediction surface back to the  affine space tangent to the original prediction surface.

\begin{figure}
	\includegraphics[height=7cm]{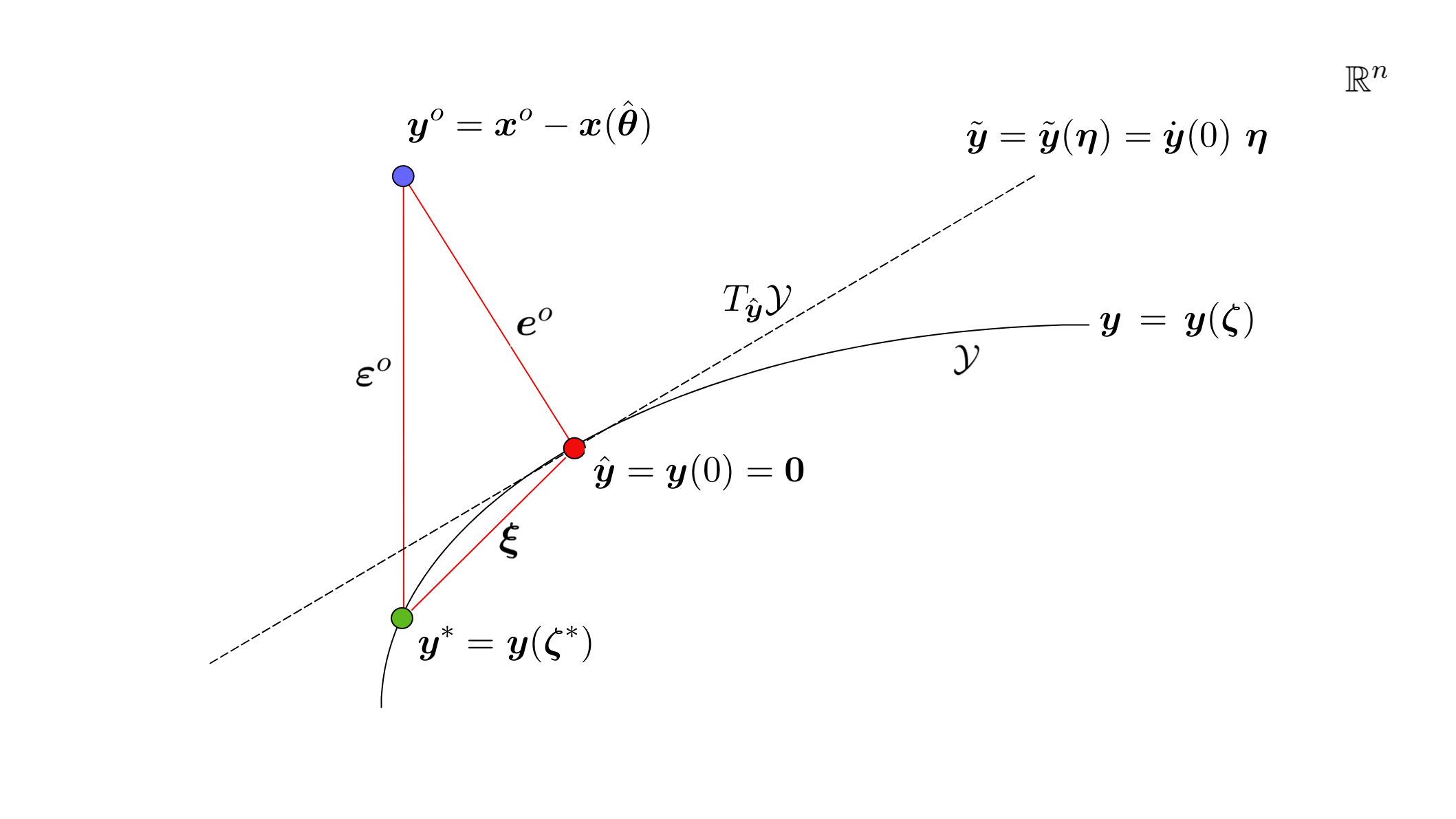}
	\caption{Nonlinear relationship in case space: reparametrized geometry}
	\label{fig:NonLinearCaseSpaceReparametrizedRn}
\end{figure}

\newpage
\subsection{The structure of the total derivative}
\label{sec:TotalDerivative}

The purpose of the present Appendix is to clarify the switch of index position, which will be used in the next Appendix \ref{sec:AffineTangentToM}.

Let $ {\boldsymbol{f}} : \mathbb{R}^m \to \mathbb{R}^n $ be a sufficiently regular (i.e.\ differentiable) function, written in components as

\begin{equation*}
		{\boldsymbol{f}}({\boldsymbol{x}}) = (f^1({\boldsymbol{x}}), \dots, f^n({\boldsymbol{x}}))^\intercal.	
\end{equation*}

The total derivative of ${\boldsymbol{f}}$ at ${\boldsymbol{x}} \in \mathbb{R}^m$ is the linear map

\begin{equation*}
		{\boldsymbol{f}}'({\boldsymbol{x}}) = D{\boldsymbol{f}}({\boldsymbol{x}}) \in \mathcal{L}(\mathbb{R}^m, \mathbb{R}^n)
\end{equation*}

satisfying the defining equation

\begin{equation}
	{\boldsymbol{f}}({\boldsymbol{x}}+{\boldsymbol{h}}) - {\boldsymbol{f}}({\boldsymbol{x}}) = {\boldsymbol{f}}'({\boldsymbol{x}}) \cdot {\boldsymbol{h}} + o(\|{\boldsymbol{h}}\|), \qquad \text{as } {\boldsymbol{h}} \to {\boldsymbol{0}}.
	\label{eq:def}
	\notag
\end{equation}

Note that we habitually use the notation ${\boldsymbol{f}}'|_{\boldsymbol{x}} := {\boldsymbol{f}}'({\boldsymbol{x}})$. 

In the scalar case $n=1$, the derivative $f'({\boldsymbol{x}})$ is an element of
$\mathcal{L}(\mathbb{R}^m, \mathbb{R})$, which one identifies with
row vectors in $\mathbb{R}^{1 \times m}$ or equivalently with the gradient
$\nabla f({\boldsymbol{x}}) \in \mathbb{R}^m$.

\medskip

By assumption of sufficient regularity, for each ${\boldsymbol{x}} \in \mathbb{R}^m$ there exists a linear map $A_{\boldsymbol{x}} : \mathbb{R}^m \to \mathbb{R}^n$
satisfying the defining equation above. Defining now ${\boldsymbol{f}}'|_{\boldsymbol{x}} = {\boldsymbol{f}}'({\boldsymbol{x}}) := A_{\boldsymbol{x}}$ yields a function

\begin{equation*}
	{\boldsymbol{f}}' : \mathbb{R}^m \longrightarrow \mathcal{L}(\mathbb{R}^m, \mathbb{R}^n),
	\qquad {\boldsymbol{x}} \longmapsto {\boldsymbol{f}}'({\boldsymbol{x}}),	
\end{equation*}

that associates to each ${\boldsymbol{x}}$ a linear transformation ${\boldsymbol{f}}'|_{\boldsymbol{x}} = {\boldsymbol{f}}'({\boldsymbol{x}})$.

Hence ${\boldsymbol{f}}'$ maps points of $\mathbb{R}^m$ into the space of linear maps
$\mathcal{L}(\mathbb{R}^m, \mathbb{R}^n)$.
When $p = 1$, this is $\mathcal{L}(\mathbb{R}^m, \mathbb{R}) \cong \mathbb{R}^{1 \times m}$,
so $f'({\boldsymbol{x}})$ may be regarded as a row vector.\\

Now fix $i \in \{1, \dots, n\}$.
Take the $i$-th component of both sides of the total derivative defining equation above:
\begin{equation*}
	f^i({\boldsymbol{x}} + {\boldsymbol{h}}) - f^i({\boldsymbol{x}})
	= \big( {\boldsymbol{f}}'({\boldsymbol{x}}) \cdot {\boldsymbol{h}} \big)^i + o(\|{\boldsymbol{h}}\|).
\end{equation*}

The map ${\boldsymbol{h}} \mapsto ({\boldsymbol{f}}'({\boldsymbol{x}}) \cdot {\boldsymbol{h}})^i$ is linear in ${\boldsymbol{h}}$ and satisfies the
defining equation for the total derivative of the scalar function $f^i$ at ${\boldsymbol{x}}$.
Hence this map is precisely the derivative of $f^i$ at ${\boldsymbol{x}}$:
\begin{equation*}
		({\boldsymbol{f}}'({\boldsymbol{x}}))^i = (f^i)'({\boldsymbol{x}}).
\end{equation*}

\subsection{Tangent affine subspace to a nonlinear manifold in  $\mathcal{H}$}
\label{sec:AffineTangentToM}

Consider the tangent space to ${\mathcal{M}_{\boldsymbol{X}}}$ at $\hat {\boldsymbol{X}}$: in general this is an affine, not a linear subspace of $\mathcal{H}$. Similarly to the procedure followed for treating a tangent space to ${\mathcal{M}_{\boldsymbol{x}}} \subset \mathbb{R}^n$ in Appendix \ref{sec:AffineTangentToS} above and denoting with ${\boldsymbol{X}}:\boldsymbol{\Theta} \in {\mathbb R}^q \to {\mathcal{H}}^n$ the model map, let ${\boldsymbol{\zeta}} = {\boldsymbol{\theta}} - \hat{{\boldsymbol{\theta}}}$ , ${\boldsymbol{Y}}({\boldsymbol{\zeta}}) := {\boldsymbol{X}}({\boldsymbol{\zeta}}+\hat{{\boldsymbol{\theta}}})-\hat{{\boldsymbol{X}}}= {\boldsymbol{X}}({\boldsymbol{\theta}})-\hat{{\boldsymbol{X}}}$. Consequently, let ${\boldsymbol{Y}}={\boldsymbol{Y}}({\boldsymbol{\zeta}}) : \boldsymbol{Z} \in {\mathbb R}^q \to \mathcal{H}$ and  $\mathcal{Y} = \{{\boldsymbol{Y}}({\boldsymbol{\zeta}}),{\boldsymbol{\zeta}}\in ({\boldsymbol{\Theta}} - \hat {\boldsymbol{\theta}}) \subset \mathbb{R}^q\}$ be respectively the reparametrized model map and the corresponding model manifold.\\

Clearly $\hat {\boldsymbol{Y}} = {\boldsymbol{Y}}(0) = {\boldsymbol{X}}(\hat{{\boldsymbol{\theta}}}) - \hat {\boldsymbol{X}} = \boldsymbol{0} \in  \mathcal{H}$ \,\,for  \,\, $0 \in \mathbb{R}^q$.\\

Notice that we can define the usual derivatives in $\mathcal{H}$ because we have a vector space structure and a metric there (conversely, it is not possible in general to define a classical derivative on the manifold itself).\\

\smallskip

We want to show $\frac{\partial {\boldsymbol{Y}}}{\partial {\boldsymbol{\zeta}}} = \frac{\partial {\boldsymbol{X}}}{\partial {\boldsymbol{\theta}}}$. In fact, \\
\begin{align*}
    \frac{\partial {\boldsymbol{Y}}}{\partial {\boldsymbol{\zeta}}} &= \frac{\partial ({\boldsymbol{X}}({\boldsymbol{\zeta}}+\hat{{\boldsymbol{\theta}}})-\hat{{\boldsymbol{X}}})}{\partial {\boldsymbol{\zeta}}}\\
    &= \frac{\partial ({\boldsymbol{X}}({\boldsymbol{\zeta}}+\hat{{\boldsymbol{\theta}}})-\hat{{\boldsymbol{X}}})}{\partial {\boldsymbol{\theta}}}\frac{\partial {\boldsymbol{\theta}}}{\partial {\boldsymbol{\zeta}}}\\
    &= \frac{\partial ({\boldsymbol{X}}({\boldsymbol{\theta}}))}{\partial {\boldsymbol{\theta}}}\frac{\partial ({\boldsymbol{\zeta}}+\hat{{\boldsymbol{\theta}}})}{\partial {\boldsymbol{\zeta}}} \\
    &= \frac{\partial {\boldsymbol{X}}({\boldsymbol{\theta}})}{\partial {\boldsymbol{\theta}}} \, .
\end{align*}

Notice that ${\dot{{\boldsymbol{Y}}}}_0 := \left. \frac{\partial {\boldsymbol{Y}}}{\partial {\boldsymbol{\zeta}}}\right|_0 : \, \Omega \rightarrow \mathbb{R}^n$ is a (generally $n$-dimensional) random variable because the differential operator maps functions with given domain and codomain to functions with the same domain and codomain, hence $\frac{\partial {\boldsymbol{Y}}}{\partial {\boldsymbol{\zeta}}} : \,  \mathbb{R}^q \to \mathcal{H}$ and $\left. \frac{\partial {\boldsymbol{Y}}}{\partial {\boldsymbol{\zeta}}}\right|_0 \in \mathcal{H}$.\\

Define $\tilde{{\boldsymbol{Y}}} : {\mathbb R}^q \to \mathcal{H},\quad \tilde{{\boldsymbol{Y}}}({\boldsymbol{\eta}})={\dot{{\boldsymbol{Y}}}}_0 \cdot {\boldsymbol{\eta}}$, \, ${\boldsymbol{\eta}} \in \mathbb{R}^q$, then
$\tilde{\mathcal{Y}} = \{\tilde{{\boldsymbol{Y}}}({\boldsymbol{\eta}})\}$ is a $q$-dimensional subspace of $\mathcal{H}$.\\

Now consider the (shifted and reparametrized) model manifold \,\,$\mathcal{Y} = \{{\boldsymbol{Y}}({\boldsymbol{\zeta}})\} = \{{\boldsymbol{X}}({\boldsymbol{\zeta}} + \hat{{\boldsymbol{\theta}}})-\hat{{\boldsymbol{X}}}, \, {\boldsymbol{\zeta}} \in {\boldsymbol{Z}} = ({\boldsymbol{\Theta}} - \hat{{\boldsymbol{\theta}}}) \subset \mathbb{R}^q\}$ \,\, as well as the original model manifold ${\mathcal{M}_{\boldsymbol{X}}}= \{{\boldsymbol{X}}({\boldsymbol{\theta}}), \, {\boldsymbol{\theta}} \in {\boldsymbol{\Theta}} \subset \mathbb{R}^q\}$. Notice that $\forall {\boldsymbol{\theta}} \in {\boldsymbol{\Theta}} \subset \mathbb{R}^q \,\, \exists \,\, {\boldsymbol{X}}({\boldsymbol{\theta}}) \in {\mathcal{M}_{\boldsymbol{X}}}$ by definition.\\

Define a map $\varphi : \mathcal{H} \rightarrow \mathcal{H},\,\, \varphi : \boldsymbol{U} \mapsto \boldsymbol{U} - {\hat{\boldsymbol{X}}}$. Notice that its restriction to ${\mathcal{M}_{\boldsymbol{X}}}$ maps the whole ${\mathcal{M}_{\boldsymbol{X}}}$ onto $\mathcal{Y}$, in fact
\begin{align*}
	\varphi \left({\boldsymbol{X}}({\boldsymbol{\theta}})\right) 
	&= {\boldsymbol{X}}({\boldsymbol{\theta}})-\hat{{\boldsymbol{X}}}\\
	&= {\boldsymbol{X}}({\boldsymbol{\theta}}-\hat{{\boldsymbol{\theta}}}+\hat{{\boldsymbol{\theta}}})-\hat{{\boldsymbol{X}}}\\
	&= {\boldsymbol{X}}({\boldsymbol{\zeta}({\boldsymbol{\theta}})} + \hat{{\boldsymbol{\theta}}}) - \hat{{\boldsymbol{X}}}\\
	&= {\boldsymbol{Y}}({\boldsymbol{\zeta}})
\end{align*}
since ${\boldsymbol{\zeta}}({\boldsymbol{\theta}})={\boldsymbol{\theta}}-\hat{{\boldsymbol{\theta}}}$ and ${\boldsymbol{Y}}({\boldsymbol{\zeta}}) = {\boldsymbol{X}}({\boldsymbol{\zeta}}+\hat{{\boldsymbol{\theta}}})-\hat{{\boldsymbol{X}}}$, so $\varphi : {\mathcal{M}_{\boldsymbol{X}}}\rightarrow \mathcal{Y}$. In other words, the following diagram commutes:

\smallskip

\begin{center}
    \begin{tikzpicture}
        \begin{scope}[every node/.style={draw=white!60,thick,draw}]
            \node (A) at (0,0) {$ {\boldsymbol{\Theta}} \subset \mathbb{R}^q$};
            \node (B) at (3,0) {${\mathcal{M}_{\boldsymbol{X}}} \subset \mathcal{H}$};
            \node (C) at (0,-2)  {${\boldsymbol{Z}} \subset \mathbb{R}^q$};
            \node (D) at (3,-2)  {$\mathcal{Y} \subset \mathcal{H}$};
        \end{scope}

        \begin{scope}[>={Stealth[black]},
            every node/.style={},
            every edge/.style={draw=black}]
            \path [->] (A) edge [above] node {${\boldsymbol{X}}$} (B);
            \path [->] (A) edge [left] node {${\boldsymbol{\zeta}}$} (C);
            \path [->] (B) edge [right] node {$\varphi$} (D);
            \path [->] (C) edge [above] node {${\boldsymbol{Y}}$} (D);
        \end{scope}
    \end{tikzpicture}
\end{center}

Notice that ${\boldsymbol{Y}}(0)={\boldsymbol{X}}(0 + \hat{{\boldsymbol{\theta}}})-\hat{{\boldsymbol{X}}} = \boldsymbol{0} \in \mathcal{H}$ as well as $\varphi(\hat{{\boldsymbol{X}}})= \boldsymbol{0} \in \mathcal{H}$, where we indicate with $0$ the zero in $\mathbb{R}^q$ and with ${\boldsymbol{0}}$ the zero random variable in $\mathcal{H}$.\\
Notice also that a model map ${\boldsymbol{X}}(\cdot)$ or ${\boldsymbol{Y}}(\cdot)$ is not a random variable itself but a function to obtain random variables from parameter values.\\
Notice finally that $\tilde{{\boldsymbol{Y}}}$ is also not a random variable, but a linear map \, $\mathbb{R}^q\rightarrow\mathcal{H}$\,: letting ${\boldsymbol{\eta}}\in\mathbb{R}^q$, we have that $\tilde{{\boldsymbol{Y}}}({\boldsymbol{\eta}})={\dot{{\boldsymbol{Y}}}}_0\cdot{\boldsymbol{\eta}}$ is a random variable and $\tilde{\mathcal{Y}}=\{\tilde{{\boldsymbol{Y}}}({\boldsymbol{\eta}}), \, {\boldsymbol{\eta}}\in\mathbb{R}^q\}$ is a subspace of $\mathcal{H}$.\\

\begin{figure}
    \includegraphics[height=7cm]{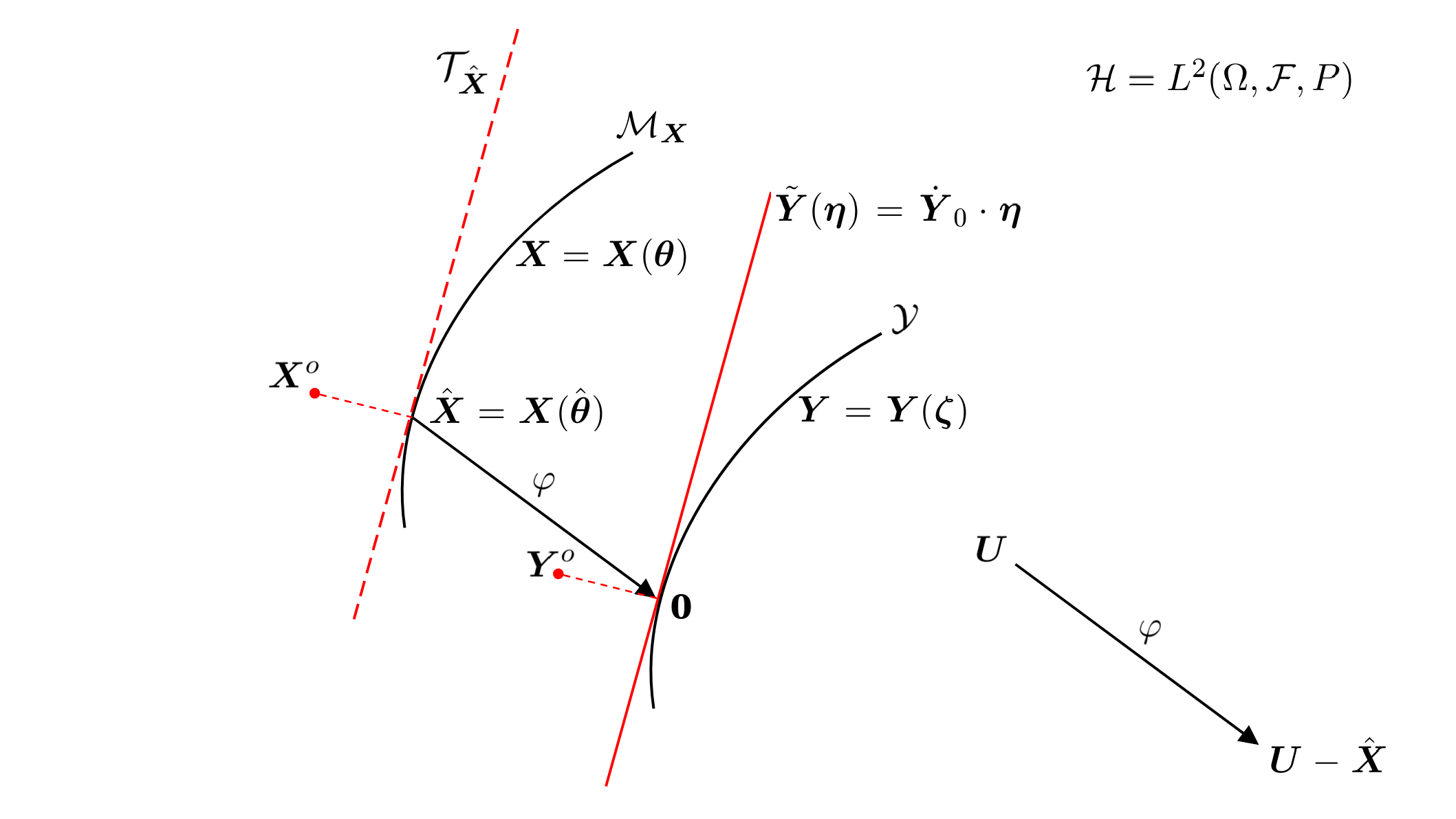}
    \caption{Nonlinear models in the Hilbert space of finite-variance random variables: original and reparametrized models}
    \label{fig:NonLinearHilbertReparametrized}
\end{figure}

While it is now clear that $\tilde{\mathcal{Y}}$ is a random variable (sub)space, it is not clear that this subspace is tangent to the reparametrized manifold $\mathcal{Y}$ at $\boldsymbol{0}$.\\

\bigskip

Recall that, given a smooth atlas $\mathcal{A}$ on ${\mathcal{M}_{\boldsymbol{X}}}$, \\
\centerline{
$\mathcal{A} = \left\{\,(x,\mathcal{U}) : \,\, (x,\mathcal{U}),(y,\mathcal{V}) \in \mathcal{A} \Rightarrow x \circ y^{-1} \in \mathcal{C}^\infty (\mathbb{R}^q) \right\} $,\\
}
a function $f : {\mathcal{M}_{\boldsymbol{X}}} \to \mathbb{R}$ is smooth, $f \in \mathcal{C}^{\infty}({\mathcal{M}_{\boldsymbol{X}}})$, if $f \circ x^{-1} : \mathbb{R}^q \to \mathbb{R}$ is smooth  $ \,\,\forall\, (x,\mathcal{U}) \in \mathcal{A}$.

\textbf{In general}, we can build a tangent space to a (differential) manifold at a point by considering the set of velocities, \textit{i.e.} of differential operators, which map functions $f \in \mathcal{C}^{\infty}(\mathcal{Y})$ ,   $f:\mathcal{Y}\rightarrow\mathbb{R}$ (WLOG, if a function to $\mathbb{R}^q$ is desired, consider it component-wise), to their directional derivatives in the direction of curves $\gamma:\mathbb{R}\rightarrow \mathcal{Y}$, $T_p\mathcal{Y} = \{\vartheta_{\gamma,p}$ such that $\, p =\gamma(\lambda)\in \mathcal{U}\subset \mathcal{Y}$ and $\forall f \in C^{\infty}(\mathcal{Y}) $ it is $\, \vartheta_{\gamma,p}(f)=(f\circ\gamma)'|_{\lambda}\}$.\\
We can define a chart-induced basis of this tangent space: we can express the directional derivative of a function $f$ in the direction of a curve $\gamma$ with respect to a specific chart map $\mu$ as the dot product of the components of the derivative of $\gamma_\mu$ at $\lambda$, denoted as \quad ${\dot{\gamma}_p} := {\dot{\gamma}_\mu}|_{\lambda} $, times the application of the operator \quad $\frac{\partial}{\partial \mu}|_p := (\cdot \circ \mu^{-1})'|_{p_\mu}$ \quad to the function $f$.\\

\begin{center}
	\begin{tikzpicture}
		\begin{scope}[every node/.style={draw=white!60,thick,draw}]
			\node (A) at (0,0)  {$\mathcal{Y}$};
			\node (B) at (-2,0) {$\mathbb{R}$};
			\node (C) at (2,0) {$\mathbb{R}$};
			\node (D) at (0,-2) {$\mathbb{R}^q$};
		\end{scope}
		
		\begin{scope}[>={Stealth[black]},
			every node/.style={},
			every edge/.style={draw=black}]
			\coordinate (Aleft) at ($(A)+(-4pt,-8pt)$);
			\coordinate (Dleft) at ($(D)+(-4pt,8pt)$);
			\coordinate (Aright) at ($(A)+(4pt,-8pt)$);
			\coordinate (Dright) at ($(D)+(4pt,8pt)$);
			\path [->] (B) edge [above] node {$\gamma$} (A);
			\path [->] (B) edge [below left] node {$\mu \circ \gamma$} (D);
			\path [->] (D) edge [below right] node {$f \circ \mu^{-1}$} (C);
			\path [->] (A) edge [above] node {$f$} (C);
			\path[->] (Aleft) edge node[left] {$\mu$} (Dleft);
			\path[->] (Dright) edge node[right] {$\mu^{-1}$} (Aright);
		\end{scope}
	\end{tikzpicture}
\end{center}

More precisely, let $\gamma:\mathbb{R}\rightarrow \mathcal{Y}$,\, $f:\mathcal{Y}\rightarrow\mathbb{R}$,\, $\mathcal{U}\subset \mathcal{Y}$, \,$p=\gamma(\lambda) \in\mathcal{U}$, \,$\mu:\mathcal{U}\rightarrow\mathbb{R}^q$, \,\,then

\begin{align*}
    (f\circ\gamma)'|_{\lambda} &= (f\circ \mu^{-1}\circ \mu\circ\gamma)'|_{\lambda}\\
    &= (f\circ \mu^{-1})'|_{\mu \circ \gamma(\lambda)}\cdot(\mu\circ\gamma)'|_{\lambda}\\
    &=(f\circ \mu^{-1})'|_{p_\mu}\cdot\, {\dot{\gamma}_\mu}|_{\lambda}\\
    &={\dot{\gamma}_\mu}|_{\lambda} \cdot\, (f\circ \mu^{-1})'|_{p_\mu}\\
    &:= {\dot{\gamma}_p}\cdot \frac{\partial}{\partial \mu}|_p(f) .
\end{align*}

Abstracting from the specific function $f$ to which the differential operator is applied, we may thus define $\vartheta_{\gamma,p} := \frac{\partial}{\partial \mu}|_p(\cdot)$ . \\

We call the curve $\gamma_\mu :\mathbb{R}\rightarrow\mathbb{R}^q$ the \underline{representation} of the curve $\gamma$ under the chart map $\mu$ and similarly $p_\mu=\mu(p)=\mu(\gamma(\lambda)) \in {\mathbb R}^q$ the \underline{representation} of $p$ under the chart map $\mu$.\\

Notice that ${\dot{\gamma}_p} := {\dot{\gamma}_\mu}|_{\lambda} $ is a $q$-dimensional scalar depending on the curve $\gamma$ and that $\frac{\partial}{\partial \mu}|_p $ is the differential operator $\frac{\partial}{\partial \mu}|_p := (\cdot \circ \mu^{-1})'|_{p_\mu}$. We are thus led to consider the set (indeed, vector space) $T_p\mathcal{Y}$ of differential operators generated by all linear combinations $\boldsymbol{g} \frac{\partial}{\partial \mu}|_p $, with $\boldsymbol{g} \in \mathbb{R}^q$. It can further be proved that $\,\forall \,\boldsymbol{g} \in \mathbb{R}^q \,\, \exists \,\gamma : \mathbb{R} \to \mathcal{Y} , \, {{\dot{\gamma}}_\mu}|_{\lambda} = \boldsymbol{g}$.

Therefore, since\\

\begin{align*}
    \boldsymbol{g} \frac{\partial}{\partial \mu}|_p &= \boldsymbol{g} \, \left.(\cdot \circ \mu^{-1})'\right|_{p_\mu}\\
    &= \sum^q_{j=1} g_j \left[\left.{\left(\cdot \circ \mu^{-1}\right)'}\right|_{p_\mu}\right]^j\\
    &= \sum^q_{j=1} g_j \left[\left.\left(\cdot \circ (\mu^{-1})^j\right)'\right|_{p_\mu}\right]
\end{align*}

(see Appendix \ref{sec:TotalDerivative} above) we have that a basis for the tangent space \,\,
$ T_p\mathcal{Y} = \{\vartheta_{\gamma,p}\} $ \,\,
is
\begin{equation*}
	\left\{\frac{\partial}{\partial \mu^j}|_p\right\}_{j=1,...,q}=\left\{\left.\left(\cdot \circ (\mu^{-1})^j\right)'\right|_{p_\mu}\right\}_{j=1,...,q}
\end{equation*}
\\

\bigskip

\textbf{In particular}, for the purpose of the present discussion, we will compute a basis for $T_p\mathcal{Y}$ induced by the canonical chart map ${\boldsymbol{Y}}^{-1}$: \\

\begin{center}
    \begin{tikzpicture}
        \begin{scope}[every node/.style={draw=white!60,thick,draw}]
            \node (A) at (0,0) {$\mathcal{Y}$};
            \node (B) at (0,-2) {$\mathbb{R}^q$};
        \end{scope}

        \begin{scope}[>={Stealth[black]},
            every node/.style={},
            every edge/.style={draw=black}]
            \path [->] (A) edge [bend right=50] [left] node {${\boldsymbol{Y}}^{-1}$} (B);
            \path [->] (B) edge [bend right=50][right] node {${\boldsymbol{Y}}$} (A);
        \end{scope}
    \end{tikzpicture}
\end{center}

Clearly we must require sufficient regularity in the model map ${\boldsymbol{X}}$, and hence ${\boldsymbol{Y}}$, to do this, such as continuous differentiability with a continuously differentiable inverse. Under these conditions it is evident, given the model map ${\boldsymbol{Y}}: \mathbb{R}^q \rightarrow {\mathcal{Y}}$, that ${\boldsymbol{Y}}^{-1} : {\mathcal{Y}} \rightarrow \mathbb{R}^q$ is a chart map, indeed a global chart map. Using as basis of the tangent space the basis induced by the chart map $({\boldsymbol{Y}}^{-1})$ and following the general development above we can write:

\begin{equation*}
    \{(\cdot \circ (({\boldsymbol{Y}}^{-1})^{-1})^j)'|_{\boldsymbol{0}}\}_{j=1,...,q}=\{(\cdot \circ {\boldsymbol{Y}}^j)'|_0\}_{j=1,...,q}
\end{equation*}

but ${\boldsymbol{Y}}'=\dot{{\boldsymbol{Y}}}$ and $({\boldsymbol{Y}}^j)'=(\dot{{\boldsymbol{Y}}})^j=\dot{{\boldsymbol{Y}}}^j$ \,\,\, hence the tangent space is the span  $T_{\boldsymbol{0}}{\mathcal{Y}} = \mathcal{S}\{ \dot{{\boldsymbol{Y}}}^j\}_{j=1,...,q}$ \,\, WRT the canonical model chart map $({\boldsymbol{Y}}^{-1})$. This means that for any ${\boldsymbol{\eta}}\in\mathbb{R}^q$ it is  $( {\dot{\boldsymbol{Y}}}_0 \cdot {\boldsymbol{\eta}})\in T_{\boldsymbol{0}}\mathcal{Y} \subset \mathcal{H}$ and ${\dot{\boldsymbol{Y}}}_0 \cdot {\boldsymbol{\eta}}$ is a random variable in the tangent space to the (reparametrized) model manifold $\mathcal{Y}$ at $\boldsymbol{0}$, as desired.\\

In other words, we have constructed a subspace $T_{\boldsymbol{0}}\mathcal{Y} \subset \mathcal{H}$ of random variables tangent to  $\mathcal{Y}$ at $\boldsymbol{0}$ using the model canonical chart map ${\boldsymbol{Y}}^{-1}$.
Since this is a subspace, the Hilbert projection theorem applies and the translated random variable ${\boldsymbol{Y}}^o={\boldsymbol{X}}^o-\hat{{\boldsymbol{X}}}$ has a unique projection onto $T_{\boldsymbol{0}}\mathcal{Y}$ at $\boldsymbol{0}$, of global minimal distance from $T_{\boldsymbol{0}}\mathcal{Y} = \tilde{\mathcal{Y}}$, such that $\forall \, {\boldsymbol{Y}} \in \tilde{\mathcal{Y}}$ we have $({\boldsymbol{Y}}^o - \boldsymbol 0) = {\boldsymbol{Y}}^o \, \perp {\boldsymbol{Y}} = {\boldsymbol{Y}}  - \boldsymbol 0$.

\bigskip

It is now trivial to shift back from the linear tangent space  $\tilde{\mathcal{Y}}$ to the affine tangent space $\mathcal{T}_{\hat {\boldsymbol{X}}}$ to ${\mathcal{M}_{\boldsymbol{X}}}$ at $\hat{{\boldsymbol{X}}}$ by means of the inverse map $\varphi^{-1} : \tilde{\mathcal{Y}} \rightarrow \mathcal{T}_{\hat {\boldsymbol{X}}},\quad \varphi^{-1} : {\boldsymbol{U}} \mapsto {\boldsymbol{U}} + \hat {\boldsymbol{X}}$, so that the original random variable ${\boldsymbol{X}}^o$ has a unique projection onto $\hat {\boldsymbol{X}} \in \mathcal{T}_{\hat {\boldsymbol{X}}}$ , of global minimal distance from $\mathcal{T}_{\hat {\boldsymbol{X}}}$, such that $\forall \, {\boldsymbol{X}} \in \mathcal{T}_{\hat {\boldsymbol{X}}}$ we have $({\boldsymbol{X}}^o - \hat {\boldsymbol{X}}) \, \perp ({\boldsymbol{X}}  - \hat {\boldsymbol{X}})$.

\medskip

This affine tangent space to the original model manifold, $\mathcal{T}_{\hat {\boldsymbol{X}}}$, is an affine space of random variables and is the best local linear approximation to ${\mathcal{M}_{\boldsymbol{X}}}$ in a neighborhood of $\hat {\boldsymbol{X}}$ in the norm of the embedding Hilbert space $\mathcal{H}$.

\section*{Acknowledgment}
The work of Prof. Andrea De Gaetano was supported by the Distinguished Professor Excellence Program of \'Obuda University, Budapest Hungary.

\bibliography{NatMathModV06}

@Book{Seber2003,
  author    = {Seber, G. A. F. and Lee, A. J.},
  publisher = {Wiley},
  title     = {Linear regression analysis},
  year      = {2003},
  edition   = {2nd},
  series    = {Wiley Series in Probability and Mathematical Statistics},
}

@Book{Seber1989,
  author    = {Seber, G. A. F. and Wild, C. J.},
  publisher = {John Wiley \& Sons},
  title     = {Nonlinear regression},
  year      = {1989},
  address   = {New York},
}

@Book{Bates1988,
  author    = {Bates, D. M. and Watts, D. G.},
  publisher = {John Wiley \& Sons},
  title     = {Nonlinear regression analysis and its applications},
  year      = {1988},
  address   = {New York},
}

@Book{Bobrowski2013,
  author    = {Bobrowski, A.},
  publisher = {Cambridge University Press},
  title     = {Functional analysis for probability and stochastic processes},
  year      = {2013},
}

@Article{BarndorffNielsen1986,
  author  = {Barndorff-Nielsen, O. E. and others},
  journal = {International Statistical Review},
  title   = {The role of differential geometry in statistical theory},
  year    = {1986},
  number  = {1},
  pages   = {83--96},
  volume  = {54},
  doi     = {10.2307/1403260},
}

@InCollection{Fletcher2020,
  author    = {Fletcher, T.},
  booktitle = {Riemannian geometric statistics in medical image analysis},
  publisher = {Academic Press},
  title     = {Statistics on manifolds},
  year      = {2020},
  editor    = {Pennec, X. and Sommer, S. and Fletcher, T.},
  isbn      = {9780128147252},
  pages     = {39--74},
  doi       = {10.1016/B978-0-12-814725-2.00009-1},
}

@Article{Nielsen2020,
  author  = {Nielsen, F.},
  journal = {Entropy},
  title   = {An elementary introduction to information geometry},
  year    = {2020},
  number  = {10},
  pages   = {1100},
  volume  = {22},
  doi     = {10.3390/e22101100},
}

@Article{Amari2010,
  author  = {Amari, S. and Cichocki, A.},
  journal = {Bulletin of the Polish Academy of Sciences, Technical Sciences},
  title   = {Information geometry of divergence functions},
  year    = {2010},
  volume  = {58},
  doi     = {10.2478/v10175-010-0019-1},
}

@Misc{Brigo2011,
  author       = {Brigo, D.},
  howpublished = {hal-00640516v3},
  title        = {The direct $L^2$ geometric structure on a manifold of probability densities with applications to filtering},
  year         = {2011},
}

@Article{Amari2009,
  author  = {Amari, S.},
  journal = {IEEE Transactions on Information Theory},
  title   = {$\alpha$-divergence is unique, belonging to both $f$-divergence and Bregman divergence classes},
  year    = {2009},
  number  = {11},
  pages   = {4925--4931},
  volume  = {55},
  doi     = {10.1109/TIT.2009.2030485},
}

@Book{Small1994,
  author    = {Small, C. G. and McLeish, D. L.},
  publisher = {John Wiley \& Sons},
  title     = {Hilbert space methods in probability and statistical inference},
  year      = {1994},
  doi       = {10.1002/9781118165522},
}

@Book{Rudin1987,
  author    = {Rudin, W.},
  publisher = {McGraw-Hill},
  title     = {Real and complex analysis},
  year      = {1987},
  edition   = {3rd},
}

@Article{Panunzi2005,
  author  = {Panunzi, S. and De Gaetano, A. and Mingrone, G.},
  journal = {American Journal of Physiology},
  title   = {Approximate linear confidence and curvature of a kinetic model of a dodecanedioic acid in humans},
  year    = {2005},
}

@Book{Folland2013,
  author    = {Folland, G. B.},
  publisher = {John Wiley \& Sons},
  title     = {Real analysis: modern techniques and their applications},
  year      = {2013},
  edition   = {2nd},
}

@Book{Rohatgi1976,
  author    = {Rohatgi, V. K.},
  publisher = {John Wiley \& Sons},
  title     = {An introduction to probability theory and mathematical statistics},
  year      = {1976},
}

@Book{Tu2010,
  author    = {Tu, L. W.},
  publisher = {Springer},
  title     = {An introduction to manifolds},
  year      = {2010},
}

@Book{Lee2009,
  author    = {Lee, J. M.},
  publisher = {American Mathematical Society},
  title     = {Manifolds and differential geometry},
  year      = {2009},
}

@Book{Lang1995,
  author    = {Lang, S.},
  publisher = {Springer},
  title     = {Differential and Riemannian manifolds},
  year      = {1995},
}
\bibliographystyle{amsplain}

\end{document}